\documentclass[onecolumn]{IEEEtran}
\usepackage[unicode, pdftex]{hyperref}
\usepackage{cite}
\usepackage{amsmath,amssymb,amsfonts,amsthm}
\usepackage{algorithmic}
\usepackage{graphicx}
\usepackage{textcomp}
\def\BibTeX{{\rm B\kern-.05em{\sc i\kern-.025em b}\kern-.08em
		T\kern-.1667em\lower.7ex\hbox{E}\kern-.125emX}}

\newcommand{\beq}[1]{\begin{equation}\label{#1}}
\newcommand{\eeq}{\end{equation}}
\newcommand{\PP}{\mathcal{P}}
\newcommand{\req}[1]{(\ref{#1})}
\newcommand{\bmu}[1]{\begin{multline}\label{#1}}
\newcommand{\emu}{\end{multline}}
\theoremstyle{definition}
\newtheorem{exmp}{Example}
\renewcommand{\(}{\left(}
\renewcommand{\)}{\right)}

\newcommand{\len}{\left\lfloor}
\newcommand{\rin}{\right\rfloor}
\newcommand{\lev}{\left\lceil}
\newcommand{\riv}{\right\rceil}

\newcommand{\eq}{\triangleq}

\newcommand{\A}{\mathcal{A}}
\renewcommand{\P}{\mathcal{P}}

\renewcommand{\L}{\mathcal{L}}
\newcommand{\x}{{\textbf{\textit{x}}}}

\newcommand{\z}{{\textbf{\textit{z}}}}

\newcommand{\e}{\mathbf{e}}

\newcommand{\p}{{\textbf{\textit{p}}}}
\newcommand{\Q}{{\mathcal Q}}

\renewcommand{\l}{\ell}

\newtheoremstyle{agdTheorem}{\parskip}{\parskip}{\itshape}{\parindent}{\bfseries}{}{0pt}{\thmname{#1}\thmnumber{~#2}.\thmnote{~\textnormal{#3.}}\quad}
\theoremstyle{agdTheorem}
\newtheorem{theorem}{Theorem}
\newtheorem{lemma}{Lemma}
\newtheorem{corol}{Corollary}
\newtheorem{proposition}{Proposition}
\newtheorem*{th::upperBoundB}{Theorem \ref{th::upperBoundB}}
\newtheorem*{th::lowerBoundB}{Theorem \ref{th::lowerBoundB}}
\newtheorem*{th::lowerBoundA}{Theorem \ref{th::lowerBoundA}}
\newtheorem*{cr::upperBoundA}{Corollary \ref{cr::upperBoundA}}
\newtheorem*{pr::capacityRate}{Proposition \ref{entropy bound}}
\newtheoremstyle{agdDefinition}{\parskip}{\parskip}{}{\parindent}{\bfseries}{}{0pt}{\thmname{#1}\thmnumber{~#2}.\thmnote{~\textnormal{#3.}}\quad}
\theoremstyle{agdDefinition}
\newtheorem{definition}{Definition}
\newtheorem{remark}{Remark}

\begin{document}
	
	\title{ Separable Codes for the \\ Symmetric
		Multiple-Access Channel}
	
	\author{
		Arkadii G. D'yachkov, Nikita A. Polyanskii,
		Ilya V. Vorobyev, and Vladislav Yu. Shchukin
		
		\thanks{A. G. D'yachkov is with the Lomonosov Moscow State University, Moscow 119991, Russia (e-mail: agd-msu@yandex.ru).}
		
		\thanks{N. A. Polyanskii is with the Israel Institute of Technology, Haifa 32000, Israel, and with  the Skolkovo Institute of Science and Technology, Moscow 121205, Russia (e-mail: nikitapolyansky@gmail.com).}
		
		\thanks{V. Yu. Shchukin is with the Institute for Information Transmission Problems, Moscow 127051, Russia (e-mail: vpike@mail.ru).}
		
		\thanks{I. V. Vorobyev is with the Skolkovo Institute of Science and Technology, Moscow 121205, Russia, and also with the Moscow Institute of Physics and Technology, Dolgoprudny 141701, Russia (e-mail: vorobyev.i.v@yandex.ru).}
		
		\thanks{A. G. D'yachkov, N. A. Polyanskii, V. Yu. Shchukin, and  I. V. Vorobyev,
			are supported by the Russian Foundation for Basic Research
			under grant No. 16-01-00440 a. N. A. Polyanskii is supported in part by the Israel Science Foundation grant nos. 1162/15, 326/17.}
	}
	
	\maketitle
	
	\begin{abstract}
		A binary matrix is called an \textit{$s$-separable code} for  the \textit{disjunctive
			multiple-access channel} (\textit{$disj$-MAC}) if Boolean sums of sets
		of $s$ columns are all distinct.  The well-known issue of the combinatorial
		coding theory is to obtain upper and lower bounds on the rate  of $s$-separable
		codes for the $disj$-MAC. In our paper, we generalize the problem  and
		discuss  upper and lower bounds on the rate  of $q$-ary
		$s$-separable  codes for  models of noiseless \textit{symmetric} MAC, i.e.,
		at each time instant the output signal
		of MAC is a symmetric function of its $s$ input signals.
	\end{abstract}
	

	\section{Introduction}
	\label{Introduction}
	\IEEEPARstart{W}{e} study some combinatorial coding problems  for the multiple access channel (MAC) that were motivated
	by two specific noiseless MAC models, corresponding to
	the transmission of $q$-ary symbols based on the frequency modulation method.
	Both  models were  suggested in
	the  paper~\cite{cw81} and were called  the $s$-user $q$-frequency MAC with (the $B$--MAC)
	and without (the $A$--MAC) intensity information.
	Using a well-known terminology~\cite{ck81} of the combinatorial coding theory, we describe
	the $A$--MAC and the $B$--MAC coding problems  along with the previously obtained results as follows.
	
	Given arbitrary integers $2\le s<t/2$, $q\ge2$ and $N\ge2$, introduce a code $X$  consisting
	of $t$  codewords of length~$N$ over a $q$-ary
	alphabet. The code $X$ is called
	\begin{itemize}
		\item
		{\em $s$-separable}~\cite{cheng2011anti} code for the $A$--MAC if for any two distinct $s$-tuples of the codewords there exists a
		coordinate $i$, $1\le i\le N$,
		in which {\em the union of $s$ elements of the  first $s$-tuple differs from the union of $s$ elements of the second $s$-tuple}.
		\item
		{\em $s$-separable}~\cite{ep16} code for the $B$--MAC if for any two distinct $s$-tuples of the codewords there exists a
		coordinate $i$, $1\le i\le N$,
		in which {\em the type (or the composition) of the  first $s$-tuple differs from the type of the second $s$-tuple}.
		\item
		{\em $(\le\!\! s)$-separable}~\cite{cheng2011anti} code for the $A$--MAC if for any $k$-tuple and any $m$-tuple,
		where $1\le k,m\le s$, of the codewords there exists a
		coordinate $i$, $1\le i\le N$, in which {\em the union of
			$k$ elements of the  $k$-tuple differs from the union of 
			$m$ elements of the  $m$-tuple}.
		\item
		{\em $s$-frameproof} code~\cite{boneh1998collusion} if for any $s$-tuple of the codewords
		and every other codeword,
		there exists a coordinate $i$, $1\le i\le N$, in which {\em the symbol of the other
			codeword doesn't belong to  the union of $s$ elements of the  $s$-tuple}.
		
		\item
		{\em $s$-hash} code~\cite{fk84,mehlhorn1984sorting} if $q\ge s$ and for every $s$-tuple of the codewords
		there exists a coordinate $i$, $1\le i\le N$, in which they all are differ.	
	\end{itemize}
	
	If $t^{(A)}(s,q,N)$ denote the largest
	size of $s$-separable  codes for the $A$--MAC, then the number
	$$
	R^{(A)}(s,q)=\varlimsup_{N\to\infty}
	\frac{\ln \,t^{(A)}(s,q,N)}{N},
	$$
	is said to be the rate of $s$-separable  codes for the $A$--MAC. By the similar way we
	define the rate $R^{(B)}(s,q)$ of $s$-separable codes for the $B$--MAC, the rate $R^{(hash)}(s,q)$ of $s$-hash codes, the rate $R^{(A)}(\le\!\!s,q)$ of $(\le\!\!s)$-separable codes
	and  the rate $R^{(fp)}(s,q)$ of $s$-frameproof codes.
	\subsection{Related Work}
	Multimedia fingerprinting is a technique to trace the sources of pirate copies of copyrighted multimedia contents. Separable codes for the $A$--MAC were introduced in~\cite{cheng2011anti} as an efficient tool to construct codes for multimedia fingerprinting in the context of ``averaging attack''. Due to its importance, constructions, applications and bounds on the rate of separable codes were further investigated and discussed in many papers~\cite{cheng2012separable,gg14, b15}. 
	
	Other security models and applications related to separable codes have been considered, and various classes of codes were defined in the literature. We only mention the most significant one and refer the reader to~\cite{boneh1998collusion}, where the problem of preventing an adversary from framing an innocent user was addressed, and the definition of frameproof codes was given. The latter were studied extensively in~\cite{ge14, cheng2011anti, staddon2001combinatorial, blackburn2003frameproof, d2017cover}. 
	
	Finally, hash codes have undergone study due to their applications in information retrieval, cryptography and algorithms. Different problems on hash codes were considered and developed in~\cite{swc08, dbk97, fk84,mehlhorn1984sorting}.
	
	Recall the well-known results emphasizing the connection between separable codes, hash codes and frameproof codes
	\beq{properties}
	\begin{split}
		R^{(A)}(\le\!\!s,q)
		&\le \min\left\{R^{(fp)}(s-1,q),\,R^{(A)}(s,q)\right\},\;   \\
		R^{(fp)}(s,q)&\le R^{(A)}(\le\!\!s,q),\\
		R^{(hash)}(s,q)&\le R^{(fp)}(s-1,q),\quad q\ge s\ge2,
	\end{split}
	\eeq
	and asymptotic ($q\to\infty$) lower and upper bounds
	\beq{asymp lower-upper bounds}
	\begin{split}
		R^{(hash)}(s,q) &\ge \frac{\ln q}{s-1}\,(1+o(1)),\\
		R^{(fp)}(s,q) &\le \frac{\ln q}{s}\,(1+o(1)).
	\end{split}
	\eeq
	
	The first and the second
	inequalities  in~\req{properties} are simple reformulations of the
	corresponding evident properties of binary superimposed codes~\cite{ks64,dr83}.
	The third  inequality  in~\req{properties} is trivially implied from the definitions.
	The upper bound for frameproof codes in~\req{asymp lower-upper bounds}  is given in~\cite{blackburn2003frameproof} and is based on the same idea as an upper bound for hash codes~\cite{dyachov1997upper,dbk97}. The asymptotic lower bound in~\req{asymp lower-upper bounds}
	is an obvious corollary of the random coding lower bound proved in~\cite{fk84,km88}.
	From~\req{properties} and \req{asymp lower-upper bounds}, it follows  the
	asymptotic ($q\to\infty$) equalities:
	\beq{known asymp equalities}
	R^{(hash)}(s,q) \sim \frac{\ln q}{s-1}, \quad	R^{(fp)}(s,q) \sim \frac{\ln q}{s}.
	\eeq
	Moreover, recent papers~\cite{gg14,b15} contains proofs of the asymptotic ($q\to\infty$) equalities:
	\beq{recent asymp equalities}
	R^{(A)}(\le\!\! 2,q) \sim \frac{2 \ln q}{3}; \quad
	R^{(A)}(\le\!\! s,q) \sim \frac{\ln q}{s-1}, \quad s \ge 3.
	\eeq
	
	Unlike \req{known asymp equalities} and \req{recent asymp equalities},
	the asymptotic behavior of the rates $R^{(A)}(s,q)$ and $R^{(B)}(s,q)$
	of $s$-separable codes for the $A$--MAC and the $B$--MAC is unknown at present.
	The aim of  our paper is a further development and generalizations
	of the given open problems. 
	
	\subsection{Outline}
	The remainder of the paper is organized as follows. After introducing notations, in Section~\ref{Statement}, we give formal definitions of MAC and a separable code for MAC, and describe
	five models of MACs,
	which are important
	for applications.  In Section~\ref{impEntropy} we discuss the entropy upper bound on the rate of separable codes for any symmetric MAC and its known and new improvements. In particular, a combinatorial upper bound on $R^{(B)}(s,q)$ is given by Theorem~\ref{th::upperBoundB}. In Section~\ref{Random A and B}, new asymptotic random coding bounds on the rate of separable codes for the $B$--MAC and the $A$--MAC are presented by Theorem~\ref{th::lowerBoundB} and Theorem~\ref{th::lowerBoundA}, respectively. In Section~\ref{LDC}, we introduce the concept of list-decoding codes for the $A$--MAC and obtain an upper bound on the rate of these codes, matching with the known lower bound for the very large alphabet size. Based on a simple connection between list-decoding codes and separable codes, we also derive an upper bound on $R^{(A)}(s,q)$, given by Corollary~\ref{cr::upperBoundA}. Finally, in Appendix, we discuss a natural probabilistic generalization of separable codes and give some  random coding bounds on the error exponent of almost separable codes  and on the rate of separable codes for any symmetric $f$--MAC.
	
	In particular,  as
	new results we claim the following.
	\begin{th::upperBoundB}
		For any $s \ge 2$ and $q \ge 2$, the rate of $s$-separable $q$-ary codes for the $B$--MAC
		satisfies the  inequality
		\begin{equation*}
		R^{(B)}(s,q)\leq
		\begin{cases}
		\frac{s+1}{2s}\ln q,\quad \text{if $s$ is odd}.\\
		\frac{s+2}{2(s+1)}\ln q,\quad \text{if $s$ is even}.\\
		\end{cases}
		\end{equation*}
	\end{th::upperBoundB}
	\begin{th::lowerBoundB}
		If $s\ge2$ is fixed and   $q\to\infty$, then
		the rate $R^{(B)}(s,q)$
		satisfies the asymptotic inequality
		$$
		R^{(B)}(s,q)\ge \frac{s}{2s-1}\ln q\,(1+o(1)).
		$$
	\end{th::lowerBoundB}
	\begin{th::lowerBoundA}
		If $s\ge2$ is fixed and   $q\to\infty$, then
		the rate $R^{(A)}(s,q)$
		satisfies the asymptotic inequality
		$$
		R^{(A)}(s,q)\ge \frac{2}{s+1}\ln q\,(1+o(1)).
		$$
	\end{th::lowerBoundA}
	\begin{cr::upperBoundA}
		For any $s \ge 2$ and $q \ge 2$, the rate of $s$-separable $q$-ary codes for the $A$--MAC
		satisfies the  inequality
		$$
		R^{(A)}(s,q)\le \frac{2}{s}\ln q.
		$$
	\end{cr::upperBoundA}


	\section{Statement of the Problem}
	\label{Statement}
	\subsection{Notations}
	\label{Notations}
	Let $q$, $N$, $t$, $s$ and $L$ be integers, where $q\ge2$, $N\ge2$, $2 \le s < t/2$, $1\le L\le t-s$;
	symbol $\eq$ is the equality by definition;
	$\A_q\eq\{0,1,\ldots,q-1\}$ is the standard $q$-ary alphabet;
	$[N] \eq \{1, 2, \ldots, N\}$ is the set of integers from $1$ to~$N$;
	$|A|$ is the size of the set~$A$;
	$\lceil b\rceil$ is the least integer $\ge b$;
	$\lfloor b\rfloor$ is the largest integer $\le b$.
	A $q$-ary $(N \times t)$-matrix $X = (x_i(j)) $, $i\in[N]$, $j\in[t]$, $x_i(j)\in\A_q$,
	with $t$ columns (\textit{codewords})
	$\x(j) \eq (x_1(j), \ldots, x_N(j))\in\A_q^N$,
	$j \in [t]$, and
	$N$ rows
	$\x_i \eq (x_i(1), x_i(2) \ldots, x_i(t))\in\A_q^t$,
	$i \in [N]$,
	is called a \textit{$q$-ary code of length $N$ and size $t$}.
	
	For a $q$-ary vector $\x = (x_1, \ldots, x_s)\eq x_1^s \in \A_q^s$,
	define the integer vector $\left(s_0,s_1,\ldots,s_{q-1}\right)$ of length $q$,
	where $s_i=s_i(\x)$,  $0\le s_i\le s$, $i\in\A_q$,
	is the number
	of positions $i$, $i\in[s]$, such that $x_i=a$. 
	Obviously, $\sum_{i=0}^{q-1}\,s_i=s$. The vector $\left(s_0,\ldots,s_{q-1}\right)$
	is said to be a {\em type}  of the $q$-ary vector $x_1^s\in\A^s_q$
	or, briefly,
	\beq{compN}
	T\left(x_1^s\right)\eq\left(s_0,\ldots,s_{q-1}\right).
	\eeq
	The set $2^{Y}$ of all subsets of a set $Y$ (or the power set of $Y$) is abbreviated by $\P(Y)$. Let $\P(Y,N)$ stand for the Cartesian product of $N$ copies of $\P(Y)$. The union of the $q$-ary vector $x_1^s\in\A_q$ is denoted by 
	\beq{uf}
	U(x_1^s)\eq \bigcup_{i\in s}x_i\in \P(\A_q).
	\eeq
	Let the standard symbol ${[t] \choose s}$ be the set of all
	$s$-subsets of the set~$[t]$. 
	For any $\e=\{e_1,\ldots,e_s\}\in {[t] \choose s}$, called a {\em message},
	and a code $X$, 
	consider the non-ordered $s$-{\it collection} of codewords 
	\beq{xe1}
	\x({\bf e})\eq\left\{\x(e_1),\ldots,\x(e_s)\right\}. 
	\eeq
	For a collection of codewords $V = \{\x(i_1),\ldots,\x(i_s)\}\subset\A_q^N$, by $T(V)$ and $U(V)$ we abbreviate the $q$-ary $(N\times q)$ matrix and the vector from $\P(\A_q, N)$ which are defined in the following way
	\beq{bch}
	\begin{split}
		T(V) \eq (T(x_1(i_1),\ldots, x_1(i_s)),\ldots, T(x_N(i_1),\ldots, x_N(i_s)))^T,\\ U(V)\eq(U(x_1(i_1),\ldots, x_1(i_s)),\ldots, U(x_N(i_1),\ldots, x_N(i_s)))^T.
	\end{split}
	\eeq
	
	\subsection{The Symmetric Multiple-Access Channel}
	\label{CombinatorialDefinitions1}
	We use
	the terminology of the noiseless (deterministic)  {\em
		multiple-access channel} (MAC), which has $s$ inputs and one
	output~\cite{ck81}. 
	Let all $s$ input alphabets of MAC be the same and coincide with the
	alphabet $\A_q$.
	Denote by $Z$ the finite output alphabet of size~$|Z|$.
	Given $s$ inputs $(x_1,\ldots, x_s)\in\A_q^s$  of MAC, the noiseless MAC is prescribed  by the function
	\beq{zf}
	z=f(x_1,\ldots,x_s)\eq f(x_1^s),\quad z\in Z,\;
	x_1^s\in\A_q^s.
	\eeq
	The  deterministic model of
	MAC is called  an~$f$--MAC.
	
	\begin{definition}
		\label{defSig}
		An $f$--MAC, given by~\req{zf}, is said to be
		the {\em symmetric} $f$--MAC if
		for any permutation  $\pi\in S_s$, where $S_s$ is the symmetric group on $s$ elements,
		the following equality holds
		\begin{equation}\label{perm}
		f\left(x_1,\ldots,x_s\right)=f\left(x_{\pi(1)},\ldots,x_{\pi(s)}\right). 
		\end{equation}
	\end{definition}
	\begin{remark}\label{rem1}
		Note that to determine a function $f=f(x_1,\ldots, x_s)=f(x_1^s)$ for the symmetric $f$--MAC it is necessary and sufficient to define $f$ only on different compositions $(s_0,s_1,\ldots, s_q)=T(x_1^s)$, $x_1^s\in\A_q^s$, or in other terms on multisets of cardinality $s$ ($s$-collections) over $A_q$. 
	\end{remark}
	In what follows, we consider
	symmetric $f$--MACs only.
	
	\subsection{Separable  Codes}
	\label{CombinatorialDefinitions2}
	For any message $\e\in {[t] \choose s}$ and a code $X$, let  $\x_i({\bf e})=\{x_i(e_1),\ldots, x_i(e_s)\}$, $i\in[N]$, be the 
	$s$-collection of signals~\req{xe1} at $s$ symmetric $f$--MAC inputs at the $i$-th time unit. 
	Then the 
	signal $z_i$, $z_i\in Z$,  $i\in[N]$,
	at the output of the symmetric $f$--MAC at the $i$-th time unit
	is
	\beq{zif}
	z_i=z^{(f)}_i({\bf e},X)\eq
	f(x_i(e_1),\ldots, x_i(e_s))\in Z. 
	\eeq
	
	On the base of the code $X$ and $N$  signals
	\beq{ChannelOut}
	\z^{(f)}({\bf e},X)\eq(z^{(f)}_1({\bf e},X),\ldots,z^{(f)}_N({\bf e},X))\in Z^N,
	\eeq
	which are known at the output of  MAC,  an {\it observer} makes the
	{\em brute force}
	decision about the unknown message~${\bf e}$.  To identify
	${\bf e}$, a code $X$ is assigned.
	\begin{definition}
		\label{Signature}
		A $q$-ary code $X$ is said to be a \textit{$s$-separable} code
		of size~$t$ and length $N$ 
		for the $f$--MAC  if  all
		$\;\z^{(f)}({\bf e},X)$, ${\;\bf e}\in {[t] \choose s}$, are distinct.
	\end{definition}
	
	
	Let $t^{(f)}(s,q,N)$ 
	be the {\it maximal  size}
	of $s$-separable $q$-ary  codes  of length $N$ for the~$f$--MAC.
	For fixed $s\ge 2$ and $q\ge2$,  the number
	\beq{Rsq}
	R^{(f)}(s,q)\eq\varlimsup_{N\to\infty}
	\frac{\ln \,t^{(f)}(s,q,N)}{N},
	\eeq
	is said to be a {\it  rate} of $s$-separable $q$-ary codes for the~$f$--MAC.
	\subsection{Examples of the Symmetric MAC}
	\label{examples}
	
	\subsubsection{$A$--MAC}
	\label{unionChannel}
	The $A$--MAC is described  by the  function
	\beq{union1}
	z=f(x_1^s)\eq\, U(x_1^s)\subseteq\A_q,
	\eeq
	where the union function $U(\cdot)$ is given in~\eqref{uf}.
	For instance, if $s=4$ and $q=3$, then
	\beq{union3}
	U(0,0,1,1)= \{0,1\},\quad
	U(1,1,0,2)= \{0,1,2\}.
	\eeq
	The cardinality $|Z|$ of  output alphabet $Z$ for the
	$A$--MAC is $|Z|=\sum\limits_{k=1}^{\min(s,q)}\,{q\choose k}$. For $s\ge q$, we have~$|Z|=2^q - 1$.
	
	\subsubsection{$B$--MAC}
	\label{compChannel}
	The  $B$--MAC known also as the compositional channel is described by the function
	\beq{comp1}
	z=f(x_1^s)\eq \,T(x_1^s),\quad x_1^s=(x_1,\ldots,x_s)\in\A_q^s, 
	\eeq
	where the type $T(\cdot)$ of a vector is defined by~\req{compN}.
	For instance, if $s=4$ and $q=3$, then
	$$
	T(0,0,1,1)=(2,2,0),\quad T(1,1,0,2)=(1,2,1).
	$$
	The cardinality of the output alphabet for the
	$B$--MAC is $|Z|=\binom{q+s-1}{s}$, $s\ge2$, $q\ge2$.
	We acknowledge  the  paper~\cite{cw81}, in which the significant
	applications of the $B$--MAC  and
	the  $A$--MAC were firstly developed. We also refer the reader
	to~\cite{cw81, bp00,wz97, ep16, vk96}, where the maximal output
	entropy of the $A$-MAC and the $B$-MAC  was investigated in different
	asymptotic and non-asymptotic cases.
	
	\subsubsection{Erasure  MAC}
	\label{erasChannel}
	A  $q$-ary  $f$--MAC is said to be the
	{\em erasure} MAC (briefly, $eras$--MAC) if it
	has  the $(q+1)$-ary  output alphabet  $Z\eq\{0,1,\ldots,q-1,*\}$ and the output function
	$z=f(x_1^s)$ has the form:
	$$
	z=f(x_1,\ldots,x_s)
	\eq\begin{cases}
	x, & \text{if $x_1=\ldots=x_s=x$},\; x\in\A_q,\cr
	*, &  \text{otherwise}.\cr
	\end{cases}
	$$
	The $eras$-MAC model  can be considered as an adequate description for the transmission
	of $q$-ary symbols based on the  {\em frequency modulation} method.
	
	\subsubsection{Threshold  MAC}
	\label{threshChannel}
	The threshold $f_{\ell}$--MAC (briefly, $\ell$-$thr$--MAC)
	has the binary input (i.e., $q=2$) and the output alphabet  $Z\eq\A_2=\{0,1\}$, and
	$$
	z= f_\ell(x_1,\ldots,x_s)
	\eq\begin{cases}
	0, & \text{if  $\sum_{i=1}^s x_i < \ell$},\cr
	1, &  \text{otherwise},\cr
	\end{cases}
	$$
	where terms of the sum are considered as $0$ and $1$ elements of the ring $\mathbb{Z}$.
	Separable codes for the $\ell$-$thr$--MAC can be used
	in \textit{compressed genotyping}~\cite{e10} models in molecular biology.
	\subsubsection{Disjunctive   MAC}
	\label{disjChannel}
	The disjunctive MAC (briefly, $disj$--MAC)
	has the binary input alphabet and the output alphabet  $Z\eq\A_2=\{0,1\}$, and
	$$
	z=f(x_1,\ldots,x_s)
	\eq\begin{cases}
	0, & \text{if  $x_1=\ldots=x_s=0$},\cr
	1, &  \text{otherwise}.\cr
	\end{cases}
	$$
	Notice that the $disj$--MAC is equivalent to the $1$-$thr$--MAC. 
	The $disj$-MAC model is interpreted   as  the transmission
	of binary symbols based on the  {\em impulse  modulation} method.
	In addition, the  binary $s$-separable codes for the $disj$-MAC are 
	closely connected with the {\em combinatorial search theory}~\cite{dh00} and
	the information-theoretic model called the {\em design of screening experiments}~\cite{d03}. 
	
	In what follows, we omit symbol $q=2$ in notations if the corresponding channel is defined only for the binary case.
	
	\section{Improvements of the Entropy Bound}\label{impEntropy}
	In this section, we first give a general statement called the entropy bound on the rate of separable codes for any symmetric MAC. For an asymptotic regime $s\to\infty$, we recall the best known bounds on the rate of separable codes for the disjunctive, the erasure, the threshold, the $A$ and the $B$  MACs in Sections~\ref{DISJSECT}-\ref{EntropyBChannel}, respectively. Finally, in Section~\ref{CUB}, we present Theorem~\ref{th::upperBoundB}, a novel upper bound, which holds for any symmetric MAC and improves the entropy bound.
	\subsection{The Entropy Upper Bound on $R^{(f)}(s,q)$ }
	Let $\p$
	be a fixed probability distribution on the alphabet $\A_q$ and the vector
	$\xi_1^s\eq\{\xi_1,\ldots,\xi_s\}$,
	$\xi_1^s\in\A_q^s$, 
	is the $s$-collection
	of {\em independent} 
	random variables having the same distribution, i.e.,
	$
	\Pr\{\xi_k=a\}\eq \p(a),\;  k\in[s],\; \,a\in\A_q
	$.
	Introduce the corresponding  Shannon entropy of the output of the symmetric $f$--MAC, i.e,
	\beq{entropy}
	H^{(f)}_{\p}(s,q)\eq
	\sum\limits_{z\in Z}\,\Pr\{f\left(\xi_1^s\right)=z\}
	\cdot\ln\frac{1}{\Pr\{f\left(\xi_1^s\right)=z\}}.
	\eeq

	The following
	statement called the {\em entropy upper bound} is a conventional information-theoretic bound.
	
	\begin{proposition}[\cite{d1989bounds}]
		\label{entropy bound}
		The rate $R^{(f)}(s,q)$ of $s$-separable  $q$-ary codes
		for the symmetric $f$--MAC satisfies the inequality
		\beq{f-up}
		R^{(f)}(s,q)\le\,C^{(f)}(s,q)\eq\frac{\max\limits_{\p}\,H_{\p}^{(f)}(s,q)}{s}.
		\eeq
	\end{proposition}
	Hereinafter, the value $C^{(f)}(s,q)$ is said to be a \textit{capacity} of $s$-separable $q$-ary codes for the $f$--MAC.
	\subsection{Bounds on the Rate $R^{(disj)}(s)$ for the Disjunctive MAC}\label{DISJSECT}
	
	One can check~\cite{m78}
	that the capacity
	of $s$-separable  binary codes for the $disj$--MAC is
	$C^{(disj)}(s)=\ln2/s$ and 
	the maximum in the right-hand
	side of~\req{f-up} is attained at the distribution~$\p$
	with probabilities  $\p(0)=2^{-1/s}$ and~$\p(1)=1-2^{-1/s}$. The significant results, improving the  corresponding  entropy
	$R^{(disj)}(s)\le \ln2/s$,
	were obtained in~\cite{cop98}  for $s=2$ and in~\cite{d14} for~$s\ge11$.
	In addition, we refer to the best known asymptotic $(s\to\infty)$ lower~\cite{d03} and
	upper~\cite{d14} bounds on the rate $R^{(disj)}(s)$:
	$$
	\frac{2(\ln2)^2}{s^2}(1+o(1))\le R^{(disj)}(s)\le \frac{4\ln s}{s^2}(1+o(1)).
	$$
	where the lower bound is based on Proposition~\ref{Lower tau(f)} formulated in Appendix.
	
	\subsection{Bounds on the Rate $R^{(eras)}(s,q)$ for the Erasure MAC}\label{ERASSECT}
	If $q=2$ and $s\to\infty$, then it is not difficult to establish~\cite{r89}
	that the capacity  of separable  $(s,2)$-codes for the $eras$--MAC is
	$C^{(eras)}(s,2)\sim \ln2/s$ and the maximum in the right-hand
	side of~\req{f-up} is asymptotically attained at distribution~$\p$
	with $\p(1)\sim\ln2/s$ or 
	with $\p(0)\sim\ln2/s$.
	In addition, we mention the best known asymptotic $(s\to\infty)$  lower~\cite{sh16} and
	upper~\cite{d03} bounds on the rate $R^{(eras)}(s,2)$:
	$$
	\frac{2(\ln2)^2}{s^2}(1+o(1))\le R^{(eras)}(s,2)\le \frac{4\ln s}{s^2}(1+o(1)).
	$$
	
	{\bf Open Problem}. In the general case $s\ge2$ and $q\ge2$, 
	we conjecture that the capacity $C^{(eras)}(s,q)$ of  the $eras$--MAC does not depend on $q\ge2$, i.e.,
	$C^{(eras)}(s,q)=C^{(eras)}(s,2)$.

	\subsection{Bounds on the Rate $R^{(\ell-thr)}(s)$ for the Threshold MAC}\label{THRSECT}
	The best known asymptotic $(\ell\ge 2$ is fixed and $s\to\infty)$  lower and
	upper bounds on the rate $R^{(\ell-thr)}(s)$ were presented in~\cite{d13,b06}:
	$$
	\frac{\ell^\ell e^{-2\ell}}{(\ell-1)!2^{\ell+1}s^2}(1+o(1))\le R^{(\ell-thr)}(s)\le \frac{2\ell^2\ln s}{s^2}(1+o(1)).
	$$
	\subsection{Bounds on the Rate $R^{(A)}(s)$ for the $A$--MAC}\label{EntropyAChannel}
	For fixed $q$ and $s\to\infty$, the best known upper bounds on the rate $R^{(A)}(s)$ are based on the upper bound for $R^{(disj)}(s)$ and improve the entropy bound. The asymptotic ($s\to\infty$) lower and upper bounds were established in~\cite{d2017cover, ge14}
	$$
	\frac{q-1 }{e\log_2 q} \frac{1}{s^2}(1+o(1))\le 	R^{(A)}(s)\le \frac{2(q-1)}{\log_2 q}\frac{\ln s}{s^2}(1+o(1)).
	$$
	\subsection{Bounds on the Rate $R^{(B)}(s)$ for the $B$--MAC}\label{EntropyBChannel}
	For fixed $q$ and $s\to\infty$, the best known lower and upper bounds on the rate $R^{(B)}(s)$ were given in~\cite{djackov1975search,d1981coding} (case $q=2$) and in~\cite{ep16} (case $q>2$)
	$$
\frac{(q-1) \ln s}{4s}(1+o(1))\le 	R^{(B)}(s)\le \frac{(q-1) \ln s}{2s}(1+o(1)).
	$$
	It is worth to note that the upper bound is actually the entropy bound, and it is quite interesting and challenging to improve it.
	\subsection{Combinatorial  Upper Bound for the  Symmetric MAC}\label{CUB}
	In the following theorem, we establish a combinatorial upper bound on the rate of $s$-separable $q$-ary codes for any symmetric $f$--MAC.
	\begin{theorem}\label{th::upperBoundB}
		For any symmetric  $f$--MAC and integers $s\ge 2$ and $q\ge 2$, the rate 
		\begin{equation}\label{upperBound}
		R^{(f)}(s,q)\overset{(a)}{\leq} R^{(B)}(s,q)\leq
		\begin{cases}
		\frac{s+1}{2s}\ln q,\quad \text{if $s$ is odd}.\\
		\frac{s+2}{2(s+1)}\ln q,\quad \text{if $s$ is even}.\\
		\end{cases}
		\end{equation}
	\end{theorem}
	Observe that inequality $(a)$ is evidently implied by Remark~\ref{rem1}. Indeed, a separable code for any symmetric $f$--MAC is also a separable code for the $B$--MAC.
	The maximal output entropy for the $B$--MAC was established in~\cite{mp78}, and it is known~\cite{cw81} that
	the capacity of $s$-separable $q$-ary codes for the $B$--MAC is
	$$
	C^{(B)}(s,q)=\sum_{s_0+\ldots+s_{q-1}=s}\frac{s!}{s_0!\ldots s_{q-1}!}\frac{1}{q^s}\ln\left(\frac{s_0!\ldots s_{q-1}!}{s!/q^s}\right).
	$$
	Therefore, $C^{(B)}(s,q)\sim \ln q$ as $q\to\infty$, and Theorem~\ref{th::upperBoundB} improves the entropy upper bound~\req{f-up} for the $B$--MAC.
	\begin{proof}[Proof of Theorem~\ref{th::upperBoundB}]
		Fix an arbitrary $q$-ary  $(N\times t)$-code~$X$.
		For any $\alpha$, $0 < \alpha < 1$, without loss of generality, we  may assume that all 
		codewords from $X$ are distinct	and the length $N$ can be represented as a sum of two integers $\alpha N$ and $(1-\alpha)N$. 
		Given $X$, introduce the bipartite graph
		$$
		G=G(X)=(V,E)\eq(V_1\cup V_2, E),\quad |V_1|=q^{\alpha N},\, |V_2|=q^{(1-\alpha) N},
		$$
		defined as follows.
		Let the vertices  in $V_1$ and $V_2$ correspond to
		distinct  $q$-ary vectors of length $\alpha N$ and $(1-\alpha) N$, respectively.
		Two vertices $v_1\in V_1$ and $v_2\in V_2$ are connected with an edge if and only if
		the code $X$ contains a codeword of  length $N=\alpha N + (1-\alpha)N$ which  is    
		the concatenation of two $q$-ary vectors corresponding to $v_1$ and~$v_2$.
		Thus, we obtain the graph $G(X)$ having $|V|=q^{(1-\alpha)N} + q^{\alpha N}$ vertices and $t$ edges, identified
		by the indexes~$[t]$ of the code $X$. In addition,
		any message ${\bf e}\in {[t] \choose s}$ is interpreted as a non-ordered $s$-collection
		of edges.
		
		Let $X$ be  a $q$-ary $s$-separable code for the $f$-MAC. Suppose, seeking a contradiction, that there exists a simple cycle $C_{2\ell}$ of length~$2\ell\le 2s$ in $G(X)$. Enumerate edges in $C_{2\ell}$ by $e_1,\ldots, e_{2\ell}$, where $e_i$ and $e_{i+1}$ are adjacent for any $i\in[2\ell-1]$ ($e_1$ and $e_{2\ell}$ are also adjacent).  Define the set $E_1$ as $\{e_1,e_3,\ldots, e_{2\ell-1}\}$, and let $E_2$ be the remaining edges of the cycle.
		Consider an arbitrary subset ${\cal S}\subset[t]\setminus\{E_1\cup E_2\}$ of the size $|{\cal S}|=s-\l$
		and define two messages $\e_i\eq E_i\cup{\cal S}\in{[t] \choose s}$, $i=1,2$.
		It is easy to check that
		outputs of the symmetric $f$-MAC for these messages 
		are the same, i.e.,
		$\z^{(f)}({\bf e}_1, X)=\z^{(f)}(\e_2, X)$.
		This contradicts to   Definition~\ref{Signature}.
		
		It is known (e.g., see \cite{nv05}) that if  a bipartite graph with two parts of sizes $n$ and $m$ does not
		contain  any simple cycle of length $\le2s$,
		then the number $t$ of its edges is 
		$$
		t\le \begin{cases}
		(2s-3)\left[(mn)^{\frac{s+1}{2s}} + m + n\right],\quad\text{if $s$ is odd}, \\
		(2s-3)\left[m^{\frac{s+2}{2s}}n^{1/2} + m + n\right],\quad\text{if $s$ is even}.
		\end{cases}
		$$
		For odd $s$, we obtain 
		$$
		t\le (2s-3)\left[q^{N\frac{s+1}{2s}}+q^{\alpha N}+q^{(1-\alpha)N}\right]\le 3(2s-3)q^{N\max\left\{\frac{s+1}{2s},\alpha, (1-\alpha)\right\}}
		$$
		Taking $\alpha=1/2$, we derive 
		$$
		t\le 3(2s-3)q^{\frac{s+1}{2s}N},
		$$
		and the rate~\req{Rsq} is upper bounded as in~\req{upperBound}.	Applying the second inequality for even $s$, we have
		$$
		t\le(2s-3)\left[q^{\frac{N}{2}(1+\frac{2\alpha}{s})} + q^{\alpha N} + q^{(1-\alpha)N}\right]\le 3(2s-3)q^{N\max\left\{\frac{s+2\alpha}{2s},\alpha, 1-\alpha\right\}}.
		$$
		Taking $\alpha$ as a root of inequality $\frac{s+2\alpha}{2s}=1-\alpha$, i.e., $\alpha = \frac{s}{2(s+1)}$, we obtain
		$$
		t\leq 3(2s-3)q^{\frac{s+2}{2(s+1)}N},
		$$
		i.e., the rate~\req{Rsq} 
		satisfies~\req{upperBound}.
	\end{proof}

	\section{Asymptotic Random Coding  Bounds for \\ the  $A$--MAC  and  the $B$--MAC}
	\label{Random A and B}
	In this section, we apply the probabilistic method  to construct asymptotic lower bounds
	on the rate of $s$-separable  $q$-ary codes for the $A$--MAC and the $B$--MAC. 
	\subsection{Random Coding Lower Bound on  $R^{(B)}(s,q)$} 
	\label{COMP}
	An asymptotic ($q\to\infty$)  random coding lower bound
	on the rate of $s$-separable  $q$-ary codes
	for the $B$--MAC is given by
	
	\begin{theorem}\label{th::lowerBoundB}
		If $s\ge2$ is fixed and   $q\to\infty$, then
		the rate $R^{(B)}(s,q)$
		satisfies the asymptotic inequality
		$$
		R^{(B)}(s,q)\ge \frac{s}{2s-1}\ln q\,(1+o(1)).
		$$
	\end{theorem}
	\begin{proof}[Proof of Theorem~\ref{th::lowerBoundB}] Consider the ensemble  of 
		matrices $X=(x_i(j))$,
		where entries $x_i(j)$, $i\in [N]$, $j\in[t]$, are  chosen independently and
		equiprobable  from the alphabet~$\A_q$. Define a \textit{bad} event $B_{j}$: ``there exist two distinct messages $\e\neq\hat\e$ from ${[t] \choose s}$ so that $j\in\e$, $j\not\in\hat\e$  and $T(\x(\e))=T(\x(\hat\e))$'', where the matrix $T(\cdot)$ is defined by~\eqref{bch}. To establish the existence of an $s$-separable $q$-ary code for the $B$--MAC, we shall upper bound the probability of the bad event
		\begin{align*}
		\Pr\{B_{j}\} =&\, \Pr\left\{\bigcup\limits_{\substack{\e,\hat\e\in{[t] \choose s}\\j\in\e,j\not\in\hat\e}} T(\x(\e))=T(\x(\hat\e))\right\}\le s \max_{m\in[s]} \Pr\left\{\bigcup\limits_{\substack{\e,\hat\e\in{[t] \choose s},\,j\in\e\\|\e\cap\hat\e|=s-m,\,j\not\in\hat\e}} T(\x(\e))=T(\x(\hat\e))\right\}\\
		\le&\,s\max_{m\in[s]} t^{2m-1}\Pr\left\{\underset{\substack{\e,\hat\e\in{[t] \choose s}\\|\e\cap\hat\e|=s-m}}{T(\x(\e))=T(\x(\hat\e))}\right\}=s\max_{m\in[s]} t^{2m-1}\left(\Pr\{T(u_1,\ldots,u_m)=T(v_1,\ldots,v_m)\}\right)^N,
		\end{align*}
		where the first and the second inequalities are evident consequences of the union bound, and $\{u_i, v_i\}|_{i=1}^m$ are independent random variables having the uniform distribution on the set $\A_q$.  Let us estimate the probability that two random $m$-tuples have the same type
		$$
		\Pr\left\{T\(u_1^m\)=T\(v_1^m\)\right\}
		=\Pr\left\{\bigcup_{\pi\in S_m}\left[\bigcap_{k=1}^m \left(u_k = v_{\pi(i)}\right)\right]\right\}
		\leq m!\cdot
		\Pr\left\{\bigcap_{k=1}^m \left(u_k = v_{\pi(k)}\right)\right\}=\frac{m!}{q^m}. 
		$$
		Therefore, 
		$$
		\Pr\{B_{j}\}\le s\max\limits_{m\in [s]}\left[t^{2m-1} (m!/q^m)^N\right].
		$$
		Since $\Pr\{B_{j}\}$ does not depend on $j\in[t]$, we deduce that if the upper bound given above is less than $1/2$, then there exists an $s$-separable $q$-ary code for the $B$--MAC of size $t/2$ and length $N$. Thus, the lower bound on $R^{(B)}(s,q)$ is as follows
		$$
		R^{(B)}(s,q)\ge\min\limits_{m\in[s]}\,\left[\frac{m\ln q-\ln m!}{2m-1}\right].
		$$
		This leads to the statement of Theorem~\ref{th::lowerBoundB}.
	\end{proof}

	\subsection{ Random Coding Lower Bound on  $R^{(A)}(s,q)$} 
	\label{JOIN}
	
	Now we establish  an asymptotic 
	random coding lower bound
	on the rate of $s$-separable $q$-ary codes
	for the $A$--MAC which is presented by
	\begin{theorem}\label{th::lowerBoundA}
		If $s\ge2$ is fixed and   $q\to\infty$, then
		the rate $R^{(A)}(s,q)$
		satisfies the asymptotic inequality
		$$
		R^{(A)}(s,q)\ge \frac{2}{s+1}\ln q\,(1+o(1)).
		$$
	\end{theorem}
	%
	%
	%
	%
	%
	\begin{proof}[Proof of Theorem~\ref{th::lowerBoundA}] Consider the ensemble  of 
		matrices $X=(x_i(j))$,
		where entries $x_i(j)$, $i\in [N]$, $j\in[t]$, are  chosen independently and
		equiprobable  from the alphabet~$\A_q$. Define a \textit{bad} event $A_{j}$: ``there exist two distinct messages $\e\neq\hat\e$ from ${[t] \choose s}$ so that $j\in\e$, $j\not\in\hat\e$  and $U(\x(\e))=U(\x(\hat\e))$'', where the vector $U(\cdot)\in\P(\A_q,N)$ is defined by~\eqref{bch}. To establish the existence of an $s$-separable $q$-ary code for the $A$--MAC, we shall upper bound the probability of the bad event
		\begin{align}
		\Pr\{A_{j}\} =&\, \Pr\left\{\bigcup\limits_{\substack{\e,\hat\e\in{[t] \choose s}\\j\in\e,j\not\in\hat\e}} U(\x(\e))=U(\x(\hat\e))\right\}\le s \max_{m\in[s]} \Pr\left\{\bigcup\limits_{\substack{\e,\hat\e\in{[t] \choose s},\,j\in\e\\|\e\cap\hat\e|=s-m,\,j\not\in\hat\e}} U(\x(\e))=U(\x(\hat\e))\right\}\notag\\
		\le&\,s\max\left[\max_{m\in\{2,\ldots, s\}} t^{s+m-1}\Pr\left\{\underset{\substack{\e,\hat\e\in{[t] \choose s}\\|\e\cap\hat\e|=s-m}}{U(\x(\e))=U(\x(\hat\e))}\right\};\,\Pr\left\{\bigcup\limits_{\substack{\e,\hat\e\in{[t] \choose s},\,j\in\e\\|\e\cap\hat\e|=s-1,\,j\not\in\hat\e}} U(\x(\e))=U(\x(\hat\e))\right\}\right],\label{probAbound}
		\end{align}
		where the first and the second inequalities are evident consequences of the union bound. For any $\e,\hat\e\in {[t] \choose s}$, $|\e\cap\hat\e|=s-m$, let us estimate the probability $\Pr\{U(\x(\e))=U(\x(\hat\e))$ as follows 
		\beq{ussb}
		\Pr\left\{\underset{\substack{\e,\hat\e\in{[t] \choose s}\\|\e\cap\hat\e|=s-m}}{U(\x(\e))=U(\x(\hat\e))}\right\} = \prod_{i=1}^N\,\Pr\left\{\underset{\substack{\e,\hat\e\in{[t] \choose s}\\|\e\cap\hat\e|=s-m}}{\bigcup_{k=1}^s\,x_i(e_k)=\bigcup_{j=1}^s\,x_i({\hat{e}_j})}\right\}
		\overset{(b)}{\le} \frac{s^{m N}}{q^{m N}},\quad   m\in [s].
		\eeq
		To prove $(b)$ in the last inequality, we employ the following fact. If $\xi_1,\ldots, \xi_{m+s}$ are independent and distributed uniformly over $\A_q$, then
		$$
		\Pr\left\{\bigcup_{k=1}^s\xi_k = \bigcup_{j=m+1}^{m+s}\xi_j\right\}\le\Pr\left\{\bigcup_{k=1}^{m}\xi_k\subset\bigcup_{i=m+1}^{m+s}\xi_i\right\}\le \left(\Pr\left\{\xi_1\in\bigcup_{i=m+1}^{m+s}\xi_i\right\}\right)^m\le \frac{s^m}{q^m}.
		$$ 
		For the second probability under the maximum in~\ref{probAbound}, we obtain an upper bound in a different way. Let $E_j$ consist of all possible pairs $(\e,\hat\e)$ so that $\e,\hat\e\in{[t] \choose s}$, $j\in\e$, $j\not\in\hat\e$ and $|\e\cap\hat\e|=s-1$.  Since $|\e \cap\hat\e|=s-1$, there exists $\hat j\in[t]$ such that $\e = \{j\}\cup \{\e \cap\hat\e\}$ and $\hat\e = \{\hat j\}\cup \{\e \cap\hat\e\}$.  For a real parameter $a$, $0<a<1$, we represent the event $\{U(\x(\e))=U(\x(\hat\e))\}$ as a disjoint union of two events. For the first one, we additionally require the Hamming distance $d_H(\cdot)$ between $\x(j)$ and $\x(\hat j)$ to be at least $a N$, i.e., $A_j(\e,\hat\e,\ge a)\eq\{U(\x(\e))=U(\x(\hat\e)), d_H(\x(j),\x(\hat j))\ge aN\}$. The remaining one is $A_j(\e,\hat\e,< a)\eq\{U(\x(\e))=U(\x(\hat\e)), d_H(\x(j),\x(\hat j))< aN\}$. Then we deal with each event individually. More concretely, 
		\begin{align*}
		\Pr\left\{\bigcup_{(\e,\hat\e)\in E_j}U(\x(\e))=U(\x(\hat\e)\right\}=&\,\Pr\left\{\bigcup_{(\e,\hat\e)\in E_j}A_j(\e,\hat\e,\ge a)\right\}+\Pr\left\{\bigcup_{(\e,\hat\e)\in E_j}A_j(\e,\hat\e,<a)\right\}\\
		\le&\,t^s \Pr\left\{\underset{(\e,\hat\e)\in E_j}{A_j(\e,\hat\e,\ge a)}\right\} + t\Pr\{d_H(\x(j),\x(\hat j)) < aN\},
		\end{align*}
		where the inequality is implied by the union bound,  and $\hat j\in[t]$, $\hat j\neq j$. Let us estimate the probability that two random $q$-ary vectors of length $N$ have the Hamming distance at most $aN$
		$$
		\Pr\{d_H(\x(j),\x(\hat j)) < aN\}=\sum\limits_{i=N - \lfloor aN\rfloor}^{N} \Pr\{d_H(\x(j),\x(\hat j)) = N - i\}=\sum\limits_{i=N - \lfloor aN\rfloor}^{N}\,{N\choose i}
		\left(\frac{1}{q}\right)^i\left(1-\frac{1}{q}\right)^{N-i}<\frac{2^N}{q^{(1-a)N}}.
		$$
		Now, for any $(\e,\hat\e)\in E_j$, we proceed with the event $A_j(\e,\hat\e,\ge a) = \{U(\x(\e))=U(\x(\hat\e)), d_H(\x(j),\x(\hat j))\ge aN\}$ as follows
		\begin{multline*}
		\Pr\{A_j(\e,\hat\e,\ge a)\} = \sum\limits_{i=0}^{N-\lceil a N\rceil}\Pr\left\{U(\x(\e))=U(\x(\hat\e)) \mid d_H(\x(j),\x(\hat j)) = N - i\right\}\Pr\left\{d_H(\x(j),\x(\hat j)) = N - i\right\} 
		\\
		\overset{(c)}{\le} \sum\limits_{i=0}^{N -\lceil a N\rceil}\,{N\choose i}
		\left(\frac{1}{q}\right)^i\left(1-\frac{1}{q}\right)^{N-i}\,\left(\frac{(s-1)^2}{q^2}\right)^{N-i}
		<\frac{\left(2s^2\right)^N}{q^{(1+a)N}}.
		\end{multline*}
		To prove $(c)$ in the last inequality, we use the following fact. If $\xi_1,\ldots, \xi_{s+1}$ are independent and distributed uniformly over $\A_q$, then
		$$
		\Pr\left\{\bigcup_{k=1}^s\xi_k = \bigcup_{j=2}^{s+1}\xi_j,\, \xi_1\neq \xi_{s+1}\right\}\le \Pr\left\{\xi_1\in \bigcup_{j=2}^{s}\xi_j,\, \xi_{s+1}\in \bigcup_{j=2}^{s}\xi_j\right\}\le\frac{(s-1)^2}{q^2}.
		$$
		Therefore,
		$$
		\Pr\left\{\bigcup_{\{\e,\hat\e\}\in E_j}U(\x(\e))=U(\x(\hat\e))\right\} \le \min_{0<a<1}\,\left[t^s\frac{\left(2s^2\right)^N}{q^{(1+a)N}}+
		t\frac{2^N}{q^{(1-a)N}}\right]
		\le\,2\min_{0<a<1}\,\left\{\max\left[\frac{t^s\left(2s^2\right)^N}{q^{(1+a)N}}\,;
		\,\frac{t\,2^N}{q^{(1-a)N}}\right]\right\}.
		$$
		Finally, summarizing the above arguments, we obtain
		$$
		\Pr\{A_j\}\le 2s\max\left[\max_{m\in\{2,\ldots, s\}}\frac{t^{s+m-1} s^{mN}}{q^{mN}}; \,\min_{0<a<1}\,\left\{\max\left[\frac{t^s\left(2s^2\right)^N}{q^{(1+a)N}}\,;
		\,\frac{t\,2^N}{q^{(1-a)N}}\right]\right\}\right].
		$$
		Since $\Pr\{A_{j}\}$ does not depend on $j\in[t]$, we deduce that if the upper bound given above is less than $1/2$, then there exists an $s$-separable $q$-ary code for the $A$--MAC of size $t/2$ and length $N$. Thus, the asymptotic $(q\to\infty)$ lower bound on $R^{(A)}(s,q)$ is as follows
		$$
		R^{(A)}(s,q)\ge\min\left[\frac{2}{s+1}; \max_{0<a<1}\,\left\{\min\left[\frac{1+a}{s};\,1-a\right]\right\}\right]\,
		\ln q\,(1+o(1)) = \frac{2}{s+1}\ln q\,(1+o(1)).
		$$
	\end{proof}
	\begin{remark}
		It is worth noticing that if we upper bound the probabilities in~\eqref{probAbound} for each $m\in[s]$ with the help of~\eqref{ussb}, then we would get only $R^{(A)}(s,q)\ge \frac{1}{s}\ln q(1+o(1))$ as $q\to\infty$. 
	\end{remark}

	\section{List Decoding Codes for the $A$--MAC}\label{LDC}
	After giving definitions and notations, in Section~\ref{Def and prop}, we derive several useful properties establishing a connection between list-decoding codes for the $A$--MAC and separable codes for the $A$--MAC and a relation between list decoding codes over alphabets of different sizes.  We recall the best known lower bounds on the rate of  list-decoding codes in Section~\ref{list-decoding codes}. Finally, we present a new combinatorial upper bound on the rate of  list-decoding codes in Section~\ref{list-decoding codes}, which also leads to an upper bound on the rate of  separable codes for the $A$--MAC.
	\subsection{Notations and Definitions}\label{Def and prop}
	Recall that $\P(\A_q,N)$ stands for the Cartesian product of $N$ copies of $\P(\A_q)$, where $\P(\A_q)$ is the set of all subsets of $\A_q$. A vector $\Q=(\Q_1,\ldots, \Q_N)^T\in\P(\A_q,N)$ is said to \textit{cover} a column
	$\x = (x_1, \dots, x_N)^T \in \A_q^N$ if $x_i \in \Q_i$ for all $i \in [N]$.
	
	\begin{definition}[\cite{sh16}]
		Given integers $s\ge1$ and  $L\ge1$,
		a $q$-ary code $X$ of size $t$ and length $N$ is said to be a \textit{list-decoding $(s,L, q)$-code}
		of size $t$ and length $N$ 
		if,
		for any $s$-collection of codewords $\{\x(j_1), \dots, \x(j_s)\}$,
		the vector $U(\x(j_1), \ldots,\x(j_s))$, defined by~\eqref{bch}, covers not more than
		$L-1$ other codewords of the code~$X$.
	\end{definition}
	
	In the case $s\ge2$ and $L = 1$, the list-decoding $(s,1,q)$-code  (or $s$-frameproof code~\cite{gg14})
	is
	an $(\le s)$-separable  $q$-ary code for the $A$--MAC. Moreover,
	list-decoding $(s,1, q)$-code provides a simple \textit{factor} decoding algorithm,
	that picks the unknown message $\e=(e_1,\dots,e_s) \in  {[t] \choose s}$ by searching
	all codewords of $X$ covered by the output signal
	\begin{align*}
	\z^{(A)}(\e, X) = U(\x(e_1),\ldots, \x(e_s))
	=\left( \bigcup\limits_{m = 1}^s \, x_1(e_m), \dots, \bigcup\limits_{m = 1}^s \, x_N(e_m) \right)^T.
	\end{align*}
	In the general case $L \ge 1$, the algorithm provides a subset of $[t]$ that contains
	$s$ elements of the message $\e$ and at most $L-1$ extra elements.
	
	Let $t(s,L, q, N)$ be the {\it maximal possible size}
	of list-decoding $(s,L, q)$-codes of length $N$.
	For fixed $s \ge 2$, $L \ge 1$ and $q \ge 2$, define a \textit{rate}
	of list-decoding $(s,L, q)$-codes:
	$$
	R(s,L, q) \eq \varlimsup_{N \to \infty} \frac{\ln t(s,L, q, N)}{N}.
	$$
	
	An important evident connection between $s$-separable  $q$-ary codes for the $A$--MAC
	and  list-decoding $(s,L,q)$-codes is formulated  as
	\begin{proposition}
		\label{signat-list}
		Any $s$-separable  $q$-ary code  for the $A$--MAC
		is a list-decoding $(s-1,2,q)$-code and, therefore,  the rate
		of $s$-separable  $q$-ary code for the $A$--MAC satisfies the inequality
		\beq{signature-list}
		R^{(A)}(s,q)\le R(s-1,2,q),\quad   s\ge2,\quad q\ge2.
		\eeq
	\end{proposition}
	Proposition~\ref{signat-list} can be seen as a simple reformulation of the
	corresponding  properties of binary list-decoding superimposed codes
	firstly introduced in~\cite{dr83}. A nontrivial recurrent inequality
	for the rate $R(s,L,q)$ of list-decoding $(s,L,q)$-codes is established by
	
	\begin{proposition}
		\label{lemQDecrease}
		For any integers $q' > q \ge 2$, $s \ge 2$ and $L \ge 1$
		the following inequality holds:
		\beq{qDecrease}
		R(s,L,q) \ge \frac{R(s,L,q')}{\lceil q' / (q - 1) \rceil}.
		\eeq
	\end{proposition}
	\begin{proof}[Proof of Proposition~\ref{lemQDecrease}]
		Assume that there exists a list-decoding $(s, L, q')$-code $X'$ of length $N$ and size $t$.
		Let $l \eq \lceil q' / (q - 1) \rceil$. Consider a $q$-ary code $C$ of length $l$ and size $l (q - 1) \ge q'$,
		which is composed from all possible codewords with one nonzero symbol:
		$$
		\begin{array}{|ccccccccccc|}
		1      & 0      & \dots  & 0      & & \dots & & q-1    & 0      & \dots  & 0      \\
		0      & 1      & \dots  & 0      & & \dots & & 0      & q-1    & \dots  & 0      \\
		\vdots & \vdots & \ddots & \vdots & & \dots & & \vdots & \vdots & \ddots & \vdots \\
		0      & 0      & \dots  & 1      & & \dots & & 0      & 0      & \dots  & q-1
		\end{array}
		$$
		Let us consider an injective map $\phi : \A_{q'} \rightarrow C$ such that $\phi(i)$ is the $(i+1)$th codeword of $C$.
		To construct
		a $q$-ary code $X$ of length $l N$ and size $t$, we replace each symbol $a\in\A_{q'}$ in all codewords in $X'$
		by $q$-ary codeword $\phi(a)$. One can easily check that the code $X$
		is a list-decoding $(s, L, q)$-code.
	\end{proof}

	\subsection{ Lower Bound on the rate $R(s,L, q)$}\label{list-decoding codes}
	In~\cite{sh16}, applying Proposition~\ref{lemQDecrease} and random coding arguments, the author established the lower bound on the rate
	of list-decoding $(s,L,q)$-codes which can be formulated as
	\begin{theorem}[{\cite[Theorem $2$]{sh16}}]
		\label{th::lowerBoundForLD}
		\textbf{1.} For any fixed $q \ge 2$, $s \ge 2$ and $L \ge 1$ the following lower bound holds:
		\beq{qLowerR}
		R(s,L, q) \ge \underline{R}(s,L,q) \eq \max \limits_{q' \ge q} \,
		\frac{- \ln P(q', s, L)}{(s + L - 1) k(q, q')},
		\eeq
		where
		\begin{align}
		\label{probExact}
		P(q, s, L) &\eq \sum_{m = 1}^{\min(q, s)} {q \choose m} \( \frac{m}{q} \)^L
		\times \sum_{k = 0}^{m} (-1)^k {m \choose k} \( \frac{m - k}{q} \)^s, \\
		\label{koefQ}
		k(q, q') &\eq
		\begin{cases}
		1, \quad &\text{for} \quad q = q',\\
		\lceil \frac{q'}{q - 1} \rceil, \quad &\text{otherwise}.
		\end{cases}
		\end{align}
		\textbf{2.} For any fixed $q \ge 2$, $L \ge 1$ and $s \to \infty$
		\beq{qLowerRAsS}
		\underline{R}(s,L,q) \ge \frac{L (q-1) (\ln2)^2}{s^2 } (1 + o(1)), \quad s \to \infty.
		\eeq
		\textbf{3.}
		For any fixed $s \ge 2$, $L \ge 1$ and $q\to\infty$, 
		\beq{LB}
		\underline{R}(s,L,q) = \frac{L}{s + L - 1}\; \ln q\; (1+o(1).
		\eeq
	\end{theorem}
	
	The lower bound $\underline{R}(s,L,q)$ defined by~\req{qLowerR}-\req{koefQ} improves the best previously known bounds
	presented in~\cite{ge14,swc08, r89} in asymptotics ($q$ is fixed,  $s \to \infty$) and in a wide range of
	parameters~$(q, s, L)$ as well.
	Some numerical results and a comparison of bounds are presented in Table~\ref{tabR23Best},
	where $q'(s, L, q)$ denotes the argument of maximum~\req{qLowerR}.
	
	\setlength{\tabcolsep}{4pt}
	\begin{table}[ht]
		\caption{The best known lower bounds on $R(s, L, q)$}
		\label{tabR23Best}
		\begin{tabular}{|l|ccccc|}
			\hline
			\multicolumn{1}{|c|}{$s$} & $2$ & $3$ & $4$ & $5$ & $6$ \\
			\hline
			$R(s, 1, 2) \ge$ & $0.1438^{1, 2, 4}$ & $0.0554^2$ & $0.0304^2$ & $0.0194^2$ & $0.0134^2$ \\
			$q'(s, 1, 2)$ & $2$ & $6$ & $7$ & $9$ & $10$ \\
			$R(s, 2, 2) \ge$ & $0.1703^2$ & $0.0799^2$ & $0.0474^2$ & $0.0316^2$ & $0.0226^2$ \\
			$q'(s, 2, 2)$ & $2$ & $6$ & $8$ & $9$ & $10$ \\
			\hline
			$R(s, 1, 3) \ge$ & $0.2939^{1, 3, 4}$ & $0.1171^{1, 4}$ & $0.0551^1$ & $0.0360^1$ & $0.0253^1$ \\
			$q'(s, 1, 3)$ & $3$ & $3$ & $8$ & $8$ & $10$ \\
			$R(s, 2, 3) \ge$ & $0.3662^1$ & $0.1583^1$ & $0.0864^1$ & $0.0585^1$ & $0.0425^1$ \\
			$q'(s, 2, 3)$ & $3$ & $3$ & $8$ & $10$ & $10$ \\
			\hline
		\end{tabular}
		\\
		$^1$ Theorem \ref{th::lowerBoundForLD}
		\qquad
		$^2$ \cite{sh16}
		\qquad
		$^3$ \cite{ge14}
		\qquad
		$^4$ \cite{swc08}
	\end{table}

	\subsection{Upper Bounds on the rates $R(s,L, q)$ and $R^{(A)}(s,q)$}
	It was also conjectured in~\cite{sh16} that the lower bound~\req{LB} is tight.
	We prove the conjecture in
	
	\begin{theorem}\label{th::upperBoundListA}
		For any $s \ge 2$, $L \ge 1$ and $q \ge 2$ the rate $R(s,L, q)$
		of list-decoding $(s,L, q)$-codes
		satisfies the  inequality
		\beq{LowBoundJoinChannel}
		R(s,L, q) \leq \frac{L}{s+L-1}\, \ln q.
		\eeq
	\end{theorem}
	In particular, Theorem~\ref{th::upperBoundListA} and Proposition~\ref{signat-list} yield to the following statement.
	\begin{corol}\label{cr::upperBoundA}
		For any $s \ge 2$ and $q \ge 2$, the rate of $s$-separable $q$-ary codes for the $A$--MAC
		satisfies the  inequality
		$$
		R^{(A)}(s,q)\le \frac{2}{s}\ln q.
		$$
	\end{corol}
	\begin{proof}[Proof of Theorem~\ref{th::upperBoundListA}]
		Consider an arbitrary code $X$ of length $N$ and size $t$.
		For a convenience  of the proof, we will use indexes $j$ $(i)$ of codewords (rows) which can exceed $t$ $(N)$,
		assuming  that the indexes are cyclically ordered, i.e.,
		\beq{proptN}
		x_n(j)= x_{n'}(j'),\quad \text{for }n-n'\equiv 0 \mod \,N,\quad j-j'\equiv 0 \mod{t}.
		\eeq
		For a codeword $\x(j)\in\A_q^N$, $j\in[t]$, by 
		$$
		\x_n^{n+L-1}(j)\eq\left(x_n(j),\dots,x_{n+L-1}(j)\right)\in\A_q^L,
		$$
		we abbreviate a {\em projection} of the codeword $\x(j)$ {\em on the coordinates} $n$, $n + 1$, \ldots, $n + L - 1$.
		A codeword $\x(j)$, $j\in[t]$, is said to be an  \textit{$L$-rare} in $X$
		if there exists a row index $n\in[N]$ such that the number of codeword indexes $j'\in[t]$, $j'\ne j$, with the same projection $\x_n^{n+L-1}(j')=\x_n^{n+L-1}(j)$ 
		is at most~$L-1$.
		Let $r=r_L(X)$ be the number of  codewords which are $L$-rare in~$X$.
		For each $L$-rare  codeword $\x(j)$, we can choose a row index $n\in [N]$,
		a  $q$-ary sequence $(a_1,\dots,a_L)\in\A_q^L$ 
		and an ordinal number (from $1$ to $L$) of the  $\x(j)$ among all $\le L$ codewords $\x(j')$, $j'\in[t]$,
		for which $\x_n^{n+L-1}(j')=\x_n^{n+L-1}(j)=(a_1,\dots,a_L)$. This correspondence is injective. 
		Therefore, the following claim  holds.
		\begin{lemma}\label{Lrare}
			For any code $X$ of length $N$, the number $r_L(X)$
			of its $L$-rare codewords   satisfies the inequality
			\beq{upLrare}
			r=r_L(X)\leq N\,L\,q^L.
			\eeq
		\end{lemma}
		Now we formulate another auxiliary statement.
		\begin{lemma}\label{Ls-rare}
			If a $q$-ary  code $X$ of length $N$ has a size
			\beq{t-cond}
			t>N\,L\,q^L\sum_{k=0}^{L-1} k!,
			\eeq
			then there exists an ordered set of codewords
			$\L_s=(\x(j_1),\ldots,\x(j_L))$ 
			such that there is no $L$-rare codeword in $\L_s$.
			In addition, for any $k\in[L-1]$, the projections of 
			$\x(j_{k})$ and $\x(j_{k+1})$
			on the coordinates $1+k(s-1),\, 2+k(s-1),\ldots ,\,L+k(s-1)$ are the same, i.e.,
			\beq{ks-proj}
			\x_{1+k(s-1)}^{L+k(s-1)}(j_k)=\x_{1+k(s-1)}^{L+k(s-1)}(j_{k+1}),\quad k\in[L-1].
			\eeq
		\end{lemma}
		\begin{proof}[Proof of Lemma~\ref{Ls-rare}]
			For any $j_1 \in [t]$, we shall try to construct a sequence
			$\L(j_1)=(\x(j_1),\x(j_2),\ldots, \x(j_L))$ of $L$ codewords by the following rules.
			The first element of the sequence $\L(j_1)$ is $\x(j_1)$.
			Let a sequence $(\x(j_1),\x(j_2),\ldots, \x(j_k))$  of length $k$, $1 \leq k \leq L$, be already constructed.
			If the last codeword $\x(j_k)$ is $L$-rare in $X$, then the process ends with a failure.
			If $k=L$ and $\x(j_L)$ is not $L$-rare in $X$, then the process successfully ends.
			Otherwise, for $k\le L-1$, we consider $L$ indexes from $1 + k(s - 1)$ to $L + k(s - 1)$.
			Since the codeword $\x(j_k)$ is not $L$-rare in $X$, we can find at least $L$ other
			codewords  with the same projection on the coordinates  from $1 + k(s - 1)$ to $L + k(s - 1)$.
			Among them there are  at most $k-1$ codewords
			that  could be already  included in the sequence $\L(j_1)$ at the previous $k-1$ steps.
			Therefore, there exists a codeword which has not been used.
			Among all such unused codewords we uniquely choose the  codeword $\x(j_{k+1})$
			with the cyclically   smallest index $j_{k+1}$ so that $j_{k+1}>j_k$
			as the $(k+1)$th element of~$\L(j_1)$. 
			\begin{exmp}
				Let $t = 4$ and indexes $j_1=2$ and $j_2=5$ are already used in constructing the sequence, i.e., the first two element of the sequence $\L(j_1)$ are $(\x(2), \x(5))$. Recall that the indexes $1,5,9,\ldots$ correspond to the codeword index $1$ as they have the same residue modulo $t=4$. Let codewords with indexes $3\, (7,11,\ldots)$ and $4\, (8,12,\ldots)$ be candidates to be the codeword at the third step. Then $7$, corresponding to $3$, is the cyclically smallest index so that $7>5$, and at the third stage we build the sequence $(\x(2), \x(5), \x(7))$.
			\end{exmp}
			
			Let us prove that there exists a codeword $\x(j_1)$ for which the described process successfully ends,
			i.e., as a result, we obtain a sequence $\L_s:=\L(j_1)$ without $L$-rare codewords.
			The process ends with a failure if and only if the  codeword $\x(j_{k+1})$ is $L$-rare at some step $k\in[L-1]$.
			Fix an arbitrary $L$-rare codeword $\x(j)$. Given $k\in L$, let $j_1$ be some element of $[t]$ so that
			we add $\x(j_k)=\x(j)$ in the sequence $\L(j_1)$ at the $k$th step.
			By construction of the sequence $\L(j_1)$ we know that the codeword $\x(j_k)$  coincides with the codeword $\x(j_{k-1})$
			on the  $L$ coordinates:
			\beq{L-coincidence}
			1+(k-1)(s-1),\, 2+(k-1)(s-1),\ldots ,
			(L-1)+(k-1)(s-1),\, L+(k-1)(s-1),
			\eeq
			and has the
			cyclically smallest index $j_k>j_{k-1}$ among all codeword indexes,
			except possibly representative indexes from $\{j_1, \ldots, j_{k-2}\}$.
			Hence,
			the codeword $\x(j_{k-1})$ is the first codeword before $\x(j_k)$, except $\x(j_1)$, \ldots $\x(j_{k-2})$,
			which has the same symbols as $\x(j_{k})$ on the  $L$ coordinates~\req{L-coincidence}.
			The number of codewords among $\x(j_1)$, \ldots, $\x(j_{k-2})$, which have the same symbols as $\x(j_{k})$ and $\x(j_{k-1})$
			on the $L$ coordinates~\req{L-coincidence}
			is from $0$ to~$k-2$.
			Therefore, for fixed codeword $\x(j)$ and position $k\in[L]$, there exist  at most $k-1$  possible options for~$\x(j_{k-1})$.
			Thus, any $L$-rare codeword $\x(j)$, uniquely chosen  as the codeword $\x(j_k)$
			in the sequence $\L_s(j_1)$,  spoils  at most $(k-1)!$  of starting codewords~$\x(j_1)$.
			In virtue of condition~\req{t-cond} and  upper bound~\req{upLrare} from Lemma~\ref{Lrare}, the code size
			$
			t > r_L(X) \cdot \sum_{k=0}^{L-1}k!.
			$
			Therefore,
			there exists a starting codeword $\x(j_1)$, such that the sequence
			$\L(j_1)$ will be successfully constructed.
		\end{proof}
		\begin{lemma}\label{s+L-1}
			For   any list-decoding $(s,L, q)$-code $X$
			of  length $N=s+L-1$, the size  $t$ of the code $X$ is upper bounded as follows:
			\beq{upBoundOnT}
			t\le (s+L-1)\,L  q^L \,\sum\limits_{k=0}^{L-1} k!.
			\eeq
		\end{lemma}
		\begin{proof}[Proof of Lemma~\ref{s+L-1}]
			Consider an arbitrary list-decoding $(s,L, q)$-code $X$ of the length~$N=s+L-1$. We prove the claim of  this lemma by contradiction. Assume that $t>(s+L-1)\,L  q^L \,\sum_{k=0}^{L-1} k!$.
			In virtue of Lemma~\ref{Ls-rare}, 
			we can
			construct the sequence 
			$\L_s=(\x(j_1),\ldots, \x(j_L))$ so that there is no $L$-rare codeword in $\L_s$, and the property~\eqref{ks-proj} holds. Let $J=\{j_1,\ldots, j_L\}$ be the set of codeword indexes. Without loss of generality, we may assume the sequence $(j_1,j_2,\ldots,j_L)$ is lexicographically ordered or $j_k<j_{k+1}$ for $k\in[L-1]$, since, otherwise, we can take~\req{proptN} $j_{k+1}$ as $j_{k+1}+t\lceil j_k/t\rceil$.
			
			Now we shall find an $s$-collection $I=\{i_1,\ldots,i_s\}\subset[t]\setminus J$ consisting of codeword indexes such that $U(\x(i_1),\ldots, \x(i_s))$ covers $L$ codewords $\{x(j),\,j\in J\}$. Recall that by covering we mean that, for any pair $(j,n)$, $j\in J$, $n\in [N]$, there exists $i\in I$ so that the symbol $x_n(j)=x_n(i)$.   Define a lexicographically ordered sequence $\PP$ of pairs so that the first $s+L-1$ pairs are from $(j_1, 1)$ to $(j_1, s+L-1)$, and the following $(s-1)(L-1)$ pairs are of the form $(j_k,n)$, where $n$ runs over all row indexes from $L+1+(k-1)(s-1)$ to $L+k(s-1)$, i.e.,
			\begin{multline*}
			\PP\eq((j_1, 1), (j_1, 2),\ldots,(j_1, L + s-1), 
			\\
			(j_2, L+1+(s-1)),(j_2, L+2+(s-1)), \ldots, (j_2, L+2(s-1)),\ldots,\\
			(j_L, L+1+(L-1)(s-1)),(j_L, L+2+(L-1)(s-1))\ldots,(j_L, sL)).
			\end{multline*}
			From~\eqref{ks-proj} it follows that if, for any pair $(j,n)$ in $\PP$,  there exists $i\in I$ so that the symbol $x_n(j)=x_n(i)$, then the $s$-collection $I$ is a required one. It remains to find appropriate $I$. Notice that the length  of $\PP$ is $sL$, and the second number in pairs goes from $1$ to $sL$. Divide the sequence $\PP$ into $s$ subsequences of length $L$ so that $\PP=(\PP_1,\ldots, \PP_s)$. Let
			$$
			\PP_k\eq((j_{k_1},(k-1)L+1),(j_{k_2},(k-1)L+2), \ldots, (j_{k_L},kL)).
			$$
			It is easy to check that the projection $\x(j_{k_L})$ (the codeword index is the same as the first number in the last pair of $\PP_k$) on the coordinates $(k-1)L+1, (k-1)L+2, \ldots, kL$ is
			\begin{align*}
			\x_{(k-1)L+1}^{kL}(j_{k_L})=\left(x_{(k-1)L+1}(j_{k_1}), x_{(k-1)L+2}(j_{k_2}), \ldots, x_{kL}(j_{k_L})\right).
			\end{align*}
			From Lemma~\ref{Ls-rare}, it follows that the codeword $\x(j_{k_L})$ is not $L$-rare. Therefore, we can
			find an index $i_k$, $i_k\not\in J$, and the corresponding
			codeword  $\x(i_k)$ such that
			the projections of $\x(i_k)$ and $\x(j_{k_L})$ on the coordinates  $(k-1)L+1, (k-1)L+2, \ldots, kL$ are the same, i.e.,
			\beq{cover}
			\x_{(k-1)L+1}^{kL}(i_k)=\x_{(k-1)L+1}^{kL}(j_{k_L}).
			\eeq			
			 Since there are $s$ subsequences $\PP_k$, which form $\PP$, we can find at most $s$ different $i_k$ so that 
			$U(\x(i_1),\ldots, \x(i_s))$ covers $L$ codewords $\{\x(j),\,j\in J\}$.
			This contradiction finishes the proof of Lemma~\ref{s+L-1}.
		\end{proof}
		Lemmas~\ref{Ls-rare} and~\ref{s+L-1} are intuitively  illustrated by the following example.
		\begin{exmp}\label{L=4s=2construction}
			Let  $L=4$, $s=2$ and~$N=L+s-1=5$. Then four  $q$-ary  codewords $\x(j_k)$, $\x(j_k)\in\A_q^5$, $k\in\{1,2,3,4\}$,
			satisfying the equalities~\req{ks-proj}
			can be written in the form:
			$$
			\begin{matrix}
			\x(j_1)= & (x_1(j_1),& x_2(j_1),& x_3(j_1),& x_4(j_1),& x_5(j_1)),\cr
			\x(j_2)= & (y_2,     & x_2(j_1),& x_3(j_1),& x_4(j_1), & x_5(j_1)),\cr
			\x(j_3)= & (y_2,     & y_3,     & x_3(j_1),& x_4(j_1), & x_5(j_1)),\cr
			\x(j_4)= & (y_2,     & y_3,     & y_4 ,    & x_4(j_1), & x_5(j_1)).\cr
			\end{matrix}
			$$
			These codewords are covered by $U(\x(i_1),\x(i_2))$, where	two $q$-ary  codewords $\x(i_1)
			,\x(i_2)\in\A_q^5$ are based on the  property~\req{cover} and   can be written in the form:
			$$
			\begin{matrix}
			\x(i_1)= & (x_1(j_1),& x_2(j_1),& x_3(j_1),& x_4(j_1), & a_1),\cr
			\x(i_2)= & (y_2,     & y_3,     & y_4,     & a_2       & x_5(j_1)).\cr
			\end{matrix}
			$$
		\end{exmp}
		To complete the proof of Theorem~\ref{th::upperBoundListA},
		consider an arbitrary list-decoding $(s,L, q)$-code $X$ of length $N$, $N>s+L-1$,
		and  size~$t$.
		Divide each codeword of the code $X$ into $s + L - 1$ parts of sizes $n_i$, where $\len\frac{N}{s + L - 1}\rin\leq n_i\leq \lev\frac{N}{s + L - 1}\riv,
		\; i\in[s+L-1]$.
		The number of different parts is upper bounded by
		$q^{\len\frac{N}{s + L - 1}\rin} + q^{\lev\frac{N}{s + L - 1}\riv}$.
		Replace each part of each codeword with a unique symbol from the $Q$-ary alphabet of
		the size~$Q \eq 2q^{\lev\frac{N}{s + L - 1}\riv}$.
		It is easy to see that the code $X'$, obtained after replacements, is a $Q$-ary list-decoding $(s, L, Q)$-code
		of  length $N=s+L-1$ and  size $t$. Thus, the inequality~\req{upBoundOnT} of Lemma~\ref{s+L-1}
		implies that the size
		$$
		t \leq (s + L - 1)L\sum\limits_{n=0}^{L-1}n!2^Lq^{L\lev\frac{N}{s + L - 1}\riv}.
		$$
		This upper bound  immediately leads to~\req{LowBoundJoinChannel}.
	\end{proof}
\appendix
%
%
In this section, we first introduce a probabilistic relaxation of separable codes called almost separable codes, and then give random coding bounds on the error exponent of almost separable codes  and on the rate of separable codes for any $f$--MAC.
\subsection{Notations and Definitions}
Given the symmetric  $f$--MAC and  a $q$-ary code $X$, a message ${\e}\in{[t] \choose s}$ is said to be  {\em bad}
for the code $X$, if there exists a message 
${\bf e'}\ne{\bf e}$ such that
$\;\z^{(f)}({\bf e'},X)=\z^{(f)}({\bf e},X)$.  
If the unknown message $\bf e$ is interpreted as the random vector taking
equiprobable values in the set ${[t] \choose s}$, then the {\em relative number}
of ``bad'' messages among all ${t\choose s}=|{[t] \choose s}|$
messages can be considered  as the  {\em error probability} $\epsilon^{(f)}(X,s)$
of  code $X$ for the {\em brute force} decoding.
\begin{definition}
	A code $X$ is called an almost $s$-separable code for the $f$--MAC with the error probability $\epsilon$ if the relative number of bad messages in the code $X$ is at most $\epsilon$, that is $\epsilon^{(f)}(X,s)\le \epsilon$.
\end{definition}
Let us introduce the classical notation of the error exponent and the capacity.
\begin{definition}
	Fix a parameter $R>0$. Define the error probability for almost $s$-separable codes
	$$
	\epsilon^{(f)}(s,q,R,N)\eq\min\limits_{X:t=\lfloor 2^{RN}\rfloor}\epsilon^{(f)}(X,s),
	$$
	where the minimum is taken over all $q$-ary codes of length $N$ and size $t$. The function
	$$
	E^{(f)}(s,q,R)\eq\varlimsup_{N\to\infty} \frac{-\log_2\epsilon^{(f)}(s,q,R,N)}{N}
	$$
	will be referred to as the error exponent for almost $s$-separable codes. The quantity
	$$
	C^{(f)}(s,q) = \sup \left\{R:\,E^{(f)}(s,q,R)>0\right\}
	$$
	is called the capacity of almost $s$-separable codes.
\end{definition}
We again emphasize that the rate of separable codes is upper bounded by the capacity of almost separable codes.
\begin{pr::capacityRate}[\cite{d1989bounds}]
The rate $R^{(f)}(s,q)$ of $s$-separable codes
for the symmetric $f$--MAC satisfies the inequality
$$
R^{(f)}(s,q)\le\,C^{(f)}(s,q)\eq\frac{\max\limits_{\p}\,H_{\p}^{(f)}(s,q)}{s},
$$
where $H_{\p}^{(f)}(s,q)$ is the Shannon entropy~\eqref{entropy} of the output of the $f$--MAC for the given input probability distribution $\p$.
\end{pr::capacityRate}
\subsection{Random Coding Error Exponent for the $f$-MAC}
Let the symbol ${\cal P}^{(f)}_N\left(s,t,(N_0,\ldots,N_{q-1})\right)$ denote the {\em average} error probability
over the {\em fixed composition  ensemble} (briefly, $FC$-ensemble) of $t$ independent
$q$-ary codewords $\x(i)$ with the same
type~$T(\x(i))=(N_0,\ldots,N_{q-1})$. By a similar  symbol ${\cal P}^{(f)}_N\left(s,t,\p\right)$ we will
denote the {\em average} error probability
over the {\em completely randomized   ensemble} (briefly, $CR$-ensemble) of
$q$-ary codes $X = (x_i(j))$ with
independent components $ x_i(j)$ having the same
distribution~$\p$, i.e.,
$\Pr\{x_i(j)=x\}\eq p(x),\;$  $i\in[N],\;$  $j\in[t],\,$ $x\in\A_q$.

Let a symmetric  $f$--MAC is represented
as the conditional probability $\tau^{(f)}(z|x_1^s)$, that is
$$
\tau^{(f)}(z|x_1^s)=\begin{cases}
1,\quad z=f(x_1^s),\\
0,\quad z\neq f(x_1^s).
\end{cases}
$$
To formulate  the results  about the logarithmic asymptotic behavior
of  probabilities ${\cal P}^{(f)}_N\left(s,t,(N_0,.,N_{q-1})\right)$
and  ${\cal P}^{(f)}_N\left(s,t,\p\right)$, we need the following auxiliary notations~\cite{d03}.
Let
\beq{tau}
\tau\eq\left\{\tau(x^s_1,z)\,:\,\tau(x^s_1,z)\ge0,\;\sum_{x^s_1,z}\tau(x^s_1,z)=1
\right\}
\eeq
be  a probability distribution on the Cartesian product
$\A_q^s\times Z$.
Using the standard symbols for the
conditional probabilities of the distribution
$\tau$, we abbreviate by 
\beq{tau-f}
\{\tau\}^{(f)}\eq\left\{\tau\,:\, \tau^{(f)}(z|x_1^s)=0\;\Rightarrow\;
\tau(z|x_1^s)=0\right\}
\eeq
the subset of probability distributions $\tau$~\req{tau}  such that the conditional probability
$\tau(z|x_1^s)=0\,$ is implied by $\tau^{(f)}(z|x_1^s)=0$.

Introduce the $\cup$-convex information-theoretic
functions of the argument $\tau\in\{\tau\}^{(f)}$:
\beq{HI}
	{\cal H}^{(f)}\left(\p,\tau\right) \eq
	\sum\limits_{x_1^s\cdot z}\tau(x_1^s,z)
	\ln\frac{\tau(x^s_1,z)}{
		\tau^{(f)}(z|x_1^s)\cdot
		\prod\limits_{k=1}^s p(x_k)},\quad
	I_m\left(\p,\tau\right) \eq \sum\limits_{x^s_1\cdot z}\tau(x^s_1,z)
	\ln\frac{\tau(x^m_1|x_{m+1}^s,z)}
	{\prod\limits_{k=1}^m p(x_k)},\; m\in[s].
\eeq
From~\req{entropy}, it follows that
the distribution
$$
\tau^{(f)}_{\p}\eq\left\{\tau^{(f)}(z|x_1^s)\cdot
\prod\limits_{k=1}^s p(x_k),\; x^s_1\in\A_q^s,\;z\in Z \right\}\in\{\tau\}^{(f)}
$$
and the functions~\req{HI} satisfy the equalities
$$
{\cal H}^{(f)}\left(\p,\tau^{(f)}_{\p}\right)=0,\quad
I_s\left(\p,\tau^{(f)}_{\p}\right)=H^{(f)}_{\p}(s,q).
$$
Now we are ready to state two random coding bounds on the error exponent $E^{(f)}(s,q,R)$.
\begin{proposition}[\cite{d03, d1989bounds}]
	\label{random coding exponent}
	Let $s\ge2$, $q\ge2$, $R>0$  be fixed  and
	the entropy  $H^{(f)}_{\p}(s,q)$
	of a fixed distribution $\p$ is defined by~$\req{entropy}$.
	If 
	code parameters $N,\,t\to\infty$ such that
	$$
	\frac{\ln t}{N}\sim R,\quad\frac{N_x}{N}\sim p(x),\;x\in\A_q,
	\quad s,\, q-const,
	$$
	then for the $FC$-ensemble
	there exists   
	\beq{FCrandomExponent}
	\lim\limits_{N\to\infty}\,
	\frac{-\ln{\cal P}^{(f)}_N\left(s,t,(N_0,.,N_{q-1})\right)}{N}
	\eq E^{(f)}_{FC}(s,q,R,\p)>0, \quad
	0<R<\frac{H^{(f)}_{\p}(s,q)}{s},
	\eeq
	and for the $CR$-ensemble
	there exists   
	\beq{CRrandomExponent}
	\lim\limits_{N\to\infty}\,\frac{-\ln {\cal P}^{(f)}_N\left(s,t,\p\right)}{N}
	\eq E^{(f)}_{CR}(s,q,R,\p)>0, \quad
	0<R<\frac{H^{(f)}_{\p}(s,q)}{s}.
	\eeq
	For any fixed~$\p$, the positive monotonically decreasing functions
	$E^{(f)}_{FC}(s,q,R,\p)$ and $E^{(f)}_{CR}(s,q,R,\p)$
	are $\cup$-convex functions of the parameter~$R>0$ of
	the following form:
	$$
	E^{(f)}_{FC}(s,q,R,\p)\eq\min\limits_{m\in[s]} E^{(f)}_{FC}(s,q,R,\p,m),
	$$
	\beq{FC-randomExponent1}
	E^{(f)}_{FC}(s,q,R,\p,m)\eq
	\min\limits_{\{\tau\}^{(f)}(\p)}\left\{{\cal H}^{(f)}(\p,\tau)+[I_m(\p,\tau)-mR]^+\right\},
	\eeq
	and
	$$
	E^{(f)}_{CR}(s,q,R,\p)\eq\min\limits_{m\in[s]} E^{(f)}_{CR}(s,q,R,\p,m),
	$$
	\beq{CR-randomExponent1}
	E^{(f)}_{CR}(s,q,R,\p,m)\eq
	\min\limits_{\{\tau\}^{(f)}}\left\{{\cal H}^{(f)}(\p,\tau)+[I_m(\p,\tau)-mR]^+\right\}.
	\eeq
	The  minimum in~$\req{FC-randomExponent1}$ is
	taken over the subset $\{\tau\}^{(f)}(\p)$ 
	of distributions $\{\tau\}^{(f)}$~$\req{tau-f}$ for which the marginal
	probabilities on $x_k$ are fixed and coincide with $p(x_k)$, $k\in[s]$,
	i.e., 
	\beq{marginalProb}
	\{\tau\}^{(f)}(\p)\eq
	\left\{\tau\,:\,\tau\in\{\tau\}^{(f)}\right.\quad \text{and} \quad
	\left.\sum_{x_1^{k-1}} \sum_{x_{k+1}^s}
	\sum_{z} \tau(x_1^s,z)=p(x_k),\; k\in[s]\right\}.
	\eeq
	The  minimum in~$\req{CR-randomExponent1}$ is
	taken over the set of all distributions~$\req{tau-f}$. In addition, for any $\p$, the error exponent of almost separable codes for the $f$--MAC
	$$
	E^{(f)}(s,q,R)\ge E_{FC}^{(f)}(s,q,R,\p),\quad E^{(f)}(s,q,R)\ge E_{CR}^{(f)}(s,q,R,\p).
	$$
\end{proposition}

\begin{remark}
	\label{rem2}
	Propositions~\ref{entropy bound},\ref{random coding exponent}
	and the  properties of  the {\em random error exponents}~\req{FCrandomExponent}
	and~\req{CRrandomExponent}                
	were formulated and proved in the papers~\cite{d1989bounds} and~\cite{d03}
	for the particular binary  case $q=2$ only.
	In the  general case $q\ge2$, we omit the proofs 
	because one can check that the given results
	are based on the same methods developed in~\cite{d1989bounds} and~\cite{d03}.
	Here we only note that for the  symmetric $f$-MAC,  definitions~\req{FC-randomExponent1}-\req{marginalProb}
	leads to the inequality
	$$
	E^{(f)}_{CR}(s,q,R,\p)\,\le\, E^{(f)}_{FC}(s,q,R,\p).
	$$
\end{remark}

Introduce the function
$$
E^{(f)}_{FC}(s,q,R)\eq\max_{\p} E^{(f)}_{FC}(s,q,R,\p)>0
$$
if $0<R<C^{(f)}(s,q)$,
where $C^{(f)}(s,q)$ is defined in the right-hand side~\req{f-up}.
Hence,  Propositions~\ref{entropy bound} and~\ref{random coding exponent}
imply  that the number
$C^{(f)}(s,q)$  can be considered as the
Shannon {\em capacity} of separable  $(s,q)$-codes for the
symmetric~$f$-MAC~\cite{mm80}.

The following
statement called the random coding lower bound on the rate $R^{(f)}(s,q)$
of $s$-separable  $q$-ary codes for the symmetric~$f$-MAC can be obtained as a
consequence of Proposition~\ref{random coding exponent}.

\begin{proposition}[\cite{d03}]
	\label{Lower tau(f)}
	The rate $R^{(f)}(s,q)$ of $s$-separable   $q$-ary codes
	for the  symmetric $f$-MAC satisfies the inequality
	$$
	R^{(f)}(s,q)\,\ge\,\underline{R}^{(f)}(s,q),\quad s\ge2,\;q\ge2,
	$$
	where for any  fixed distribution $\p$ the lower bound $\underline{R}^{(f)}(s,q)$
	can be represented in the form
	\beq{FC-lower}
	\underline{R}^{(f)}(s,q)\eq
	\min\limits_{m\in[s]}\,
	\frac{E^{(f)}_{FC}(s,q,0,\p,m)}{s+m-1}
	=\min\limits_{m\in[s]}\,
	\frac{\min\limits_{\{\tau\}^{(f)}(\p)}\left\{{\cal H}^{(f)}(\p,\tau)+I_m(\p,\tau)\right\}}{s+m-1}
	\eeq
	or in the form
	\beq{CR-lower}
	\underline{R}^{(f)}(s,q)\eq
	\min\limits_{m\in[s]}\,
	\frac{E^{(f)}_{CR}(s,q,0,\p,m)}{s+m-1}
	=\min\limits_{m\in[s]}\,
	\frac{\min\limits_{\{\tau\}^{(f)}}\left\{{\cal H}^{(f)}(\p,\tau)+I_m(\p,\tau)\right\}}{s+m-1}
	\eeq
\end{proposition}
In paper~\cite{d03}, Proposition~\ref{Lower tau(f)} was proved for
the particular case of the $B$-MAC with  binary $(q=2)$ alphabet only.
For the arbitrary symmetric~$f$-MAC, one can use the same arguments.
The asymptotic lower bound on the rate $R^{(disj)}(s)$ for the disjunctive MAC
formulated in~Sect.~\ref{DISJSECT} was actually  obtained in~\cite{d03} as a nontrivial 
consequence of Proposition~\ref{Lower tau(f)}. 
	\bibliographystyle{ieeetr}
	\bibliography{separable}

\begin{thebibliography}{10}

\bibitem{cw81}
S.~C. Chang and J.~K. Wolf, ``On the {$T$}-user {$M$}-frequency noiseless
  multiple-access channel with and without intensity information,'' {\em IEEE
  Trans. Inform. Theory}, vol.~27, no.~1, pp.~41--48, 1981.

\bibitem{ck81}
I.~Csiszar and J.~K{\"o}rner, {\em Information theory: coding theorems for
  discrete memoryless systems}.
\newblock Cambridge University Press, 2011.

\bibitem{cheng2011anti}
M.~Cheng and Y.~Miao, ``On anti-collusion codes and detection algorithms for
  multimedia fingerprinting,'' {\em IEEE transactions on information theory},
  vol.~57, no.~7, pp.~4843--4851, 2011.

\bibitem{ep16}
E.~Egorova and V.~Potapova, ``Signature codes for a special class of multiple
  access channel,'' in {\em Problems of Redundancy in Information and Control
  Systems (REDUNDANCY), 2016 XV International Symposium}, pp.~38--42, IEEE,
  2016.

\bibitem{boneh1998collusion}
D.~Boneh and J.~Shaw, ``Collusion-secure fingerprinting for digital data,''
  {\em IEEE Transactions on Information Theory}, vol.~44, no.~5,
  pp.~1897--1905, 1998.

\bibitem{fk84}
M.~L. Fredman and J.~Koml\'os, ``On the size of separating systems and families
  of perfect hash functions,'' {\em SIAM J. Algebraic Discrete Methods},
  vol.~5, no.~1, pp.~61--68, 1984.

\bibitem{mehlhorn1984sorting}
K.~Mehlhorn, ``Sorting and searching, volume 1 of data structures and
  algorithms,'' 1984.

\bibitem{cheng2012separable}
M.~Cheng, L.~Ji, and Y.~Miao, ``Separable codes,'' {\em IEEE Transactions on
  Information Theory}, vol.~58, no.~3, pp.~1791--1803, 2012.

\bibitem{gg14}
F.~Gao and G.~Ge, ``New bounds on separable codes for multimedia
  fingerprinting,'' {\em IEEE Trans. Inform. Theory}, vol.~60, no.~9,
  pp.~5257--5262, 2014.

\bibitem{b15}
S.~R. Blackburn, ``Probabilistic existence results for separable codes,'' {\em
  IEEE Trans. Inform. Theory}, vol.~61, no.~11, pp.~5822--5827, 2015.

\bibitem{ge14}
C.~Shangguan, X.~Wang, G.~Ge, and Y.~Miao, ``New bounds for frameproof codes,''
  {\em IEEE Trans. Inform. Theory}, vol.~63, no.~11, pp.~7247--7252, 2017.

\bibitem{staddon2001combinatorial}
J.~N. Staddon, D.~R. Stinson, and R.~Wei, ``Combinatorial properties of
  frameproof and traceability codes,'' {\em IEEE transactions on information
  theory}, vol.~47, no.~3, pp.~1042--1049, 2001.

\bibitem{blackburn2003frameproof}
S.~R. Blackburn, ``Frameproof codes,'' {\em SIAM Journal on Discrete
  Mathematics}, vol.~16, no.~3, pp.~499--510, 2003.

\bibitem{d2017cover}
A.~G. D’yachkov, I.~V. Vorobyev, N.~Polyanskii, and V.~Y. Shchukin,
  ``Cover-free codes and separating system codes,'' {\em Designs, Codes and
  Cryptography}, vol.~82, no.~1-2, pp.~197--209, 2017.

\bibitem{swc08}
D.~R. Stinson, R.~Wei, and K.~Chen, ``On generalized separating hash
  families,'' {\em J. Combin. Theory Ser. A}, vol.~115, no.~1, pp.~105--120,
  2008.

\bibitem{dbk97}
L.~Bassalygo, M.~Burmester, A.~Dyachkov, and G.~Kabatianskii, ``Hash codes,''
  in {\em Proc. IEEE Int'l Symp. Inf. Theory (ISIT)}, pp.~174--174, 1997.

\bibitem{ks64}
W.~Kautz and R.~Singleton, ``Nonrandom binary superimposed codes,'' {\em IEEE
  Trans. Inform. Theory}, vol.~10, no.~4, pp.~363--377, 1964.

\bibitem{dr83}
A.~D'yachkov and V.~Rykov, ``A survey of superimposed code theory,'' {\em
  Problems Control Inform. Theory/Problemy Upravlen. Teor. Inform.}, vol.~12,
  no.~4, pp.~229--242, 1983.

\bibitem{dyachov1997upper}
A.~Dyachkov, ``An upper bound for hash codes,'' in {\em Conference ``Computer
  Science and Information Technologies''}, pp.~219--221, 1997.

\bibitem{km88}
J.~K\"orner and K.~Marton, ``New bounds for perfect hashing via information
  theory,'' {\em European J. Combin.}, vol.~9, no.~6, pp.~523--530, 1988.

\bibitem{bp00}
L.~A. Bassalygo and M.~S. Pinsker, ``Evaluation of the asymptotics of the
  summarized capacity of an $m$-frequency $t$-user noiseless multiple-access
  channel,'' {\em Probl. Inf. Trans.}, vol.~36, no.~2, pp.~91--97, 2000.

\bibitem{wz97}
L.~Wilhelmsson and K.~Zigangirov, ``On the asymptotic capacity of a
  multiple-access channel,'' {\em Probl. Inf. Trans.}, vol.~33, no.~1,
  pp.~9--16, 1997.

\bibitem{vk96}
A.~H. Vinck and K.~J. Keuning, ``On the capacity of the asynchronous $t$-user
  $m$-frequency noiseless multiple-access channel without intensity
  information,'' {\em IEEE Trans. Inform. Theory}, vol.~42, no.~6,
  pp.~2235--2238, 1996.

\bibitem{e10}
Y.~Erlich, A.~Gordon, M.~Brand, G.~J. Hannon, and P.~P. Mitra, ``Compressed
  genotyping,'' {\em IEEE Trans. Inform. Theory}, vol.~56, no.~2, pp.~706--723,
  2010.

\bibitem{dh00}
D.-Z. Du and F.~K. Hwang, {\em Combinatorial group testing and its
  applications}, vol.~12 of {\em Series on Applied Mathematics}.
\newblock World Scientific Publishing Co., Inc., River Edge, NJ, second~ed.,
  2000.

\bibitem{d03}
A.~D'yachkov, {\em Lectures on Designing Screening Experiments}, vol.~10 of
  {\em Lect. Note Ser.}
\newblock Pohang, Korea: Pohang Univ. of Science and Technology (POSTECH),
  2003.

\bibitem{d1989bounds}
A.~G. D'yachkov and A.~Rashad, ``Universal decoding for random design of
  screening experiments,'' {\em Microelectronics and Reliability}, vol.~29,
  no.~6, pp.~965--971, 1989.

\bibitem{m78}
M.~Malyutov, ``The separating property of random matrices,'' {\em Mathematical
  Notes}, vol.~23, no.~1, pp.~84--91, 1978.

\bibitem{cop98}
D.~Coppersmith and J.~B. Shearer, ``New bounds for union-free families of
  sets,'' {\em Electron. J. Combin.}, vol.~5, pp.~Research Paper 39, 16, 1998.

\bibitem{d14}
A.~G. D'yachkov, I.~V. Vorob'ev, N.~A. Polyansky, and V.~Y. Shchukin, ``Bounds
  on the rate of disjunctive codes,'' {\em Probl. Inf. Transm.}, vol.~50,
  no.~1, pp.~27--56, 2014.
\newblock Translation of Problemy Peredachi Informatsii {{\bf{5}}0} (2014), no.
  1, 31--63.

\bibitem{r89}
A.~M. Rashad, ``On symmetrical superimposed codes,'' {\em J. Inform. Process.
  Cybernet.}, vol.~25, no.~7, pp.~337--341, 1989.

\bibitem{sh16}
V.~Y. Shchukin, ``List decoding for a multiple access hyperchannel,'' {\em
  Probl. Inf. Trans.}, vol.~52, no.~4, pp.~329--343, 2016.

\bibitem{d13}
A.~D'yachkov, V.~Rykov, C.~Deppe, and V.~Lebedev, ``Superimposed codes and
  threshold group testing,'' in {\em Information theory, combinatorics, and
  search theory}, vol.~7777 of {\em Lecture Notes in Comput. Sci.},
  pp.~509--533, Springer, Heidelberg, 2013.

\bibitem{b06}
A.~De~Bonis and U.~Vaccaro, ``Optimal algorithms for two group testing
  problems, and new bounds on generalized superimposed codes,'' {\em IEEE
  Trans. Inform. Theory}, vol.~52, no.~10, pp.~4673--4680, 2006.

\bibitem{djackov1975search}
A.~Dyachkov, ``On a search model of false coins,'' in {\em Topics in
  Information Theory (ed. Csiszar, I.-Elias, P.). Colloqua Mathematica
  Sociatotis Janos Bolyai}.

\bibitem{d1981coding}
A.~G. D'yachkov and V.~V. Rykov, ``On a coding model for a multiple-access
  adder channel,'' {\em Problemy Peredachi Informatsii}, vol.~17, no.~2,
  pp.~26--38, 1981.

\bibitem{mp78}
P.~Mateev, ``On the entropy of the multinomial distribution,'' {\em Theory of
  Probability \& Its Applications}, vol.~23, no.~1, pp.~188--190, 1978.

\bibitem{nv05}
A.~Naor and J.~Verstra\"ete, ``A note on bipartite graphs without
  {$2k$}-cycles,'' {\em Combin. Probab. Comput.}, vol.~14, no.~5-6,
  pp.~845--849, 2005.

\bibitem{mm80}
M.~B. Malyutov and P.~S. Mateev, ``Planning of screening experiments for a
  nonsymmetric response function,'' {\em Mathematical notes of the Academy of
  Sciences of the USSR}, vol.~27, no.~1, pp.~57--68, 1980.

\end{thebibliography}

\end{document}